\documentclass[journal]{IEEEtran}

\usepackage{graphicx}
\usepackage{subfigure}
\usepackage{amssymb}
\usepackage{amsbsy}
\usepackage{ulem}
\usepackage{upgreek}
\usepackage{siunitx}
\usepackage{algorithm}
\usepackage{algorithmicx}
\usepackage{algpseudocode}
\usepackage{threeparttable}
\usepackage[cmex10]{amsmath}
\usepackage{array}
\usepackage{caption}
\usepackage{url}
\usepackage{epstopdf}
\usepackage{verbatim}
\usepackage{color}
\usepackage{float}
\usepackage{wrapfig}
\usepackage{xspace}
\usepackage{listings}

\usepackage{enumerate}
\usepackage[T1]{fontenc}
\newcommand*{\etc}{
    \@ifnextchar{.}
        {etc}
        {etc.\@\xspace}
}
\graphicspath{{figures/}}

\newcommand{\refsec}[1]{$\S$\ref{sec:#1}}
\newcommand{\etal}[0]{\textit{et al.} }

\begin{document}

\newtheorem{thm}{Theorem}
\newtheorem{lma}{Lemma}
\newtheorem{defi}{Definition}
\newtheorem{proper}{Property}
\newtheorem{propo}{Proposition}

\title{A Multiple-Entanglement Routing Framework for Quantum Networks}

\author{Tu N. Nguyen,
        Kashyab J. Ambarani,
        Linh Le, Ivan Djordjevic, and Zhi-Li Zhang
\thanks{T. N. Nguyen and L. Le  are with Kennesaw State University, Marietta, GA 30060, USA.}
\thanks{K. J. Ambarani is with the Department of Computer Science, Purdue University Fort Wayne, IN 46805 USA.}
\thanks{I. Djordjevic is with the Department of Electrical and Computer Engineering, University of Arizona, Tucson, AZ 85721, USA.}
\thanks{Z.-L. Zhang is with the Department of Computer Science \& Engineering, University of Minnesota, Minneapolis, MN 55455, USA.}
\thanks{Corresponding author: Tu N. Nguyen (e-mail: tu.nguyen@kennesaw.edu).}

}

\IEEEcompsoctitleabstractindextext{%
\begin{abstract}

Quantum networks are gaining momentum in finding applications in a wide range of domains.
However, little research has investigated the potential of a quantum network framework to enable highly reliable communications. The goal of this work is to investigate and design the \textit{multiple entanglement routing} framework, namely $k$-entangled routing. In particular, the $k$-entangled routing will enable $k$ paths connecting all demands (source-destination pairs) in the network.
To design the $k$-entangled routing, we propose two algorithms that are called Sequential Multi-path Scheduling Algorithm and Min-Cut-based Multi-path Scheduling Algorithm.
In addition, we evaluate the performance of the proposed algorithms and models through a realistic quantum network simulator, NetSquid, \cite{netsquid,Coopmans2021} that models the stochastic processes underlying quantum communications. The results show that the proposed algorithms (SMPSA and MCSA) largely enhance the network's traffic flexibility. The proposed paradigms would lay the foundation for further research on the area of entanglement routing.

\end{abstract}
}

\maketitle
\IEEEdisplaynotcompsoctitleabstractindextext
\IEEEpeerreviewmaketitle

\section{Introduction} \label{sec:intro}

With ever more devices connected to the Internet and new services created, privacy and security concerns are growing rapidly along with the development of network technologies.
To address the network security issues, quantum communication serves as one of the most promising technologies for the forthcoming quantum revolution \cite{Kimble2008, Kozlowski2019,Chung2021IllinoisEQ}. The novel networking paradigm has already given rise to a variety of new applications, which are provably impossible to construct using the classical communication which powers today's internet. Quantum key distribution (QKD) \cite{Lo2014, Biham2006} $-$ one of the most famous examples $-$ promises to achieve unconditional security in communication using the laws of quantum mechanics. Along with QKD, applications such as quantum cloud computing \cite{Castelvecchi2017, Xin2018} and clock synchronization \cite{Ilo-Okeke2018, Ben2011} are also expected to become widespread with the adoption of quantum networks. Central to the aforementioned applications is that a quantum network allows for the transmission of quantum bits (qubits). The qubits possess several unique properties as compared to classical bits, which make them well-suited candidates for security applications \cite{Lindblad1999, Guo2001}. However, these unique properties also give rise to the challenge of long-distance communications within quantum networks \cite{Gisin2007}. The bottleneck of distant communication occurs due to the exponential increase in entanglement failure on increasing the physical distance between two respective quantum nodes \cite{Briegel1998}. To tackle the issue of long distance entanglement generation, quantum repeaters, i.e., devices implementing the entanglement swapping protocol, were developed  \cite{Briegel1998}. Yet, with the advent of quantum repeaters, the need for efficient and scalable protocols for long distance quantum communication (within the network layer of the quantum network stack \cite{Dahlberg2019}) has further intensified. On account of the growing need for such protocols, several studies have investigated the various challenges within the network layer \cite{Das2018, Gyongyosi2019, Kozlowski2020, Dai2020}. This paper focuses particularly on the problem of entanglement routing. The objective of entanglement routing is to establish long distance entanglement over multiple hops of quantum repeaters through entanglement swapping \cite{Chakraborty2019}. Entanglement routing serves as the basis for the successful operation of the network layer, and presents exceptional challenges due to the convoluted and stochastic nature of the physical phenomena underlying entanglement \cite{Caleffi2017}. Extensive research efforts have examined the entanglement routing problem under a variety of scenarios \cite{Caleffi2017, Pant2019, Shi2020, Chakraborty2020, Cicconetti2021}.

\textbf{Motivation.} Despite the abundance of entanglement routing algorithms proposed via recent works, to our knowledge there are no current studies which investigate the $k$-entangled routing problem. The author in \cite{Caleffi2017} analyzes the construction of an optimal routing metric, taking into account the various complex physical mechanisms underlying entanglement such as the attenuation length of optical fiber, the atomic bell state measurement efficiency, and the atom pulse duration. However, the proposed routing scheme experiences scalability issues over arbitrary topologies due to the time complexity. Pant \etal \cite{Pant2019} discusses solutions for multi-path entanglement routing within grid networks, Das \etal \cite{Das2018} studies distinct topologies applicable for entanglement routing, and Pirandola \etal \cite{Pirandola2017} constructs a multi-path routing paradigm for diamond networks. Clearly, these studies simplify the complexity of the routing problem by only considering \textit{fixed} network topologies. In addition, despite the pragmatic offline and online algorithm proposed by Shi \etal \cite{Shi2020}, the demand scheduling aspect of multi-path routing so as to maximize resource (qubit) utilization is not considered. Furthermore, Cicconetti \etal \cite{Cicconetti2021} examines the demand scheduling aspect of entanglement routing using various queuing policies. However, only single path routing is enabled via the proposed algorithms.

\textbf{Our contributions.} Going beyond the shortcomings of previous works,
this work contributes to the development of quantum networks, especially aspects of entanglement routing. It implements a key scientific principle of entanglement routing that incorporates design, analysis, and intervention and outcome measurements of the $k$-entangled routing framework in a tightly coupled and dynamic loop using realistic topologies and traffic metrics, with the goal of improving the traffic flexibility. Specifically, this paper has the following intellectual merits and contributions:

\begin{itemize}
    \item \textit{Formulation.} We mathematically formulate the $k$-entangled routing problem through an optimization framework, and briefly describe the rationale, practical implications, and hardness of the problem.
    \item \textit{Design.} We propose two novel multi-path scheduling algorithms to address the $k$-entangled routing problem, each providing distinct optimality and time complexity. Particularly, we present (i) Sequential Multi-Path Scheduling Algorithm (SMPSA) to ensure fair sharing of resources among incoming demands, and (ii) Min-Cut-based Multi-Path Scheduling Algorithm (MCSA) which provides greater efficiency and prioritizes resources towards demands with greater flexibility constraints. Moreover, we conduct rigorous theoretical analysis to better characterize the performance of the proposed paradigms.
    \item \textit{Analysis and Assessment.} We perform experiments using the NetSquid simulator \cite{Coopmans2021} to setup realistic quantum network environments and better model the stochastic processes underlying entanglement routing. Furthermore, these results are incorporated into meticulously designed simulations to compare the proposed algorithms in terms of their (i) scheduling efficiency, (ii) path selection efficiency, and (iii) resource utilization efficiency.
\end{itemize}

\textbf{Organization.} The paper is organized as follows.
In $\S$\ref{sec:primary}, we present the network model including the main components used to construct a quantum network, along with the basic notations and preliminary concepts. Furthermore, the research problem definition, associated formulation, and discussion of hardness is also outlined within this section.
We present the multi-path scheduling algorithms, theoretical analysis, and experiments to better model the efficiency of the proposed algorithms in $\S$\ref{sec:algos}. In $\S$\ref{sec:eval}, we evaluate the performance of the proposed algorithms.
Finally, we make key concluding remarks and provide possible future research directions in $\S$\ref{sec:conclusion}.

\section{Network Model and Research Problem} \label{sec:primary}

In \refsec{components}, we provide a background of the quantum network architecture, describing the main components and their functionalities. Since the architecture of quantum networks is \textit{bold} and \textit{ambitious} $-$ its
design and development will require multi-year efforts $-$ we then present the general topology, basic mathematical notations, and fundamental concepts used throughout the paper in \refsec{topology}. In \refsec{problem}, we present our vision for the proposed $k$-entangled routing framework and provide an overview of the research problem, together with an understanding of the hardness of the proposed research problem and a number of key innovative abstractions for realizing the envisaged $k$-entangled routing problem.

\subsection{Quantum Network: Main Components} \label{sec:components}

We consider a quantum network that is built on top of the underlying conventional network system. The system model presented in this paper follows the architecture of the existing physical experiments and studies \cite{Pant2019} to illustrate a practical quantum network. A quantum network is constituted by the following critical components:

\textit{1) Quantum computer} that is connected to the other quantum computers via quantum channels to establish end-to-end entanglement and form a quantum network. Each quantum computer is equipped with a quantum processor, memory, and communication qubits, and is capable of performing quantum entanglement and teleportation. Furthermore, all quantum computers are connected to the classical Internet, and are therefore able to interchange classical information as well \cite{Kozlowski2020}.

\textit{2) Quantum repeater} has been experimentally demonstrated as one of the most promising technologies towards forming long-distance entanglements
\footnote{A long-distance entanglement decays exponentially with the physical distance between the two entangled nodes \cite{Yin2017}. Quantum teleportation over a distance of 143 km has been deployed but still in the early stage.} between two quantum computers \cite{Munro2015}. To enable long-distance quantum communication, quantum repeaters execute the entanglement-swapping and entanglement-purification protocols \cite{Behera2019}. We remark that every quantum computer includes capabilities of a quantum repeater, and thereafter we refer to both of them as \textit{quantum nodes}.

\textit{3) Quantum channel} is initialized when the entanglement between two quantum nodes is established. In particular, a quantum network uses conventional optical fiber to transmit the \textit{correlation} of qubits during the entanglement establishment process, upon which a quantum link is set on the respective quantum channel.
Furthermore, we provide the architecture of a quantum network as illustrated in Fig. \ref{fig:net-architecture}.

\begin{figure}[t]
\center
\subfigure{\includegraphics[width=8cm]{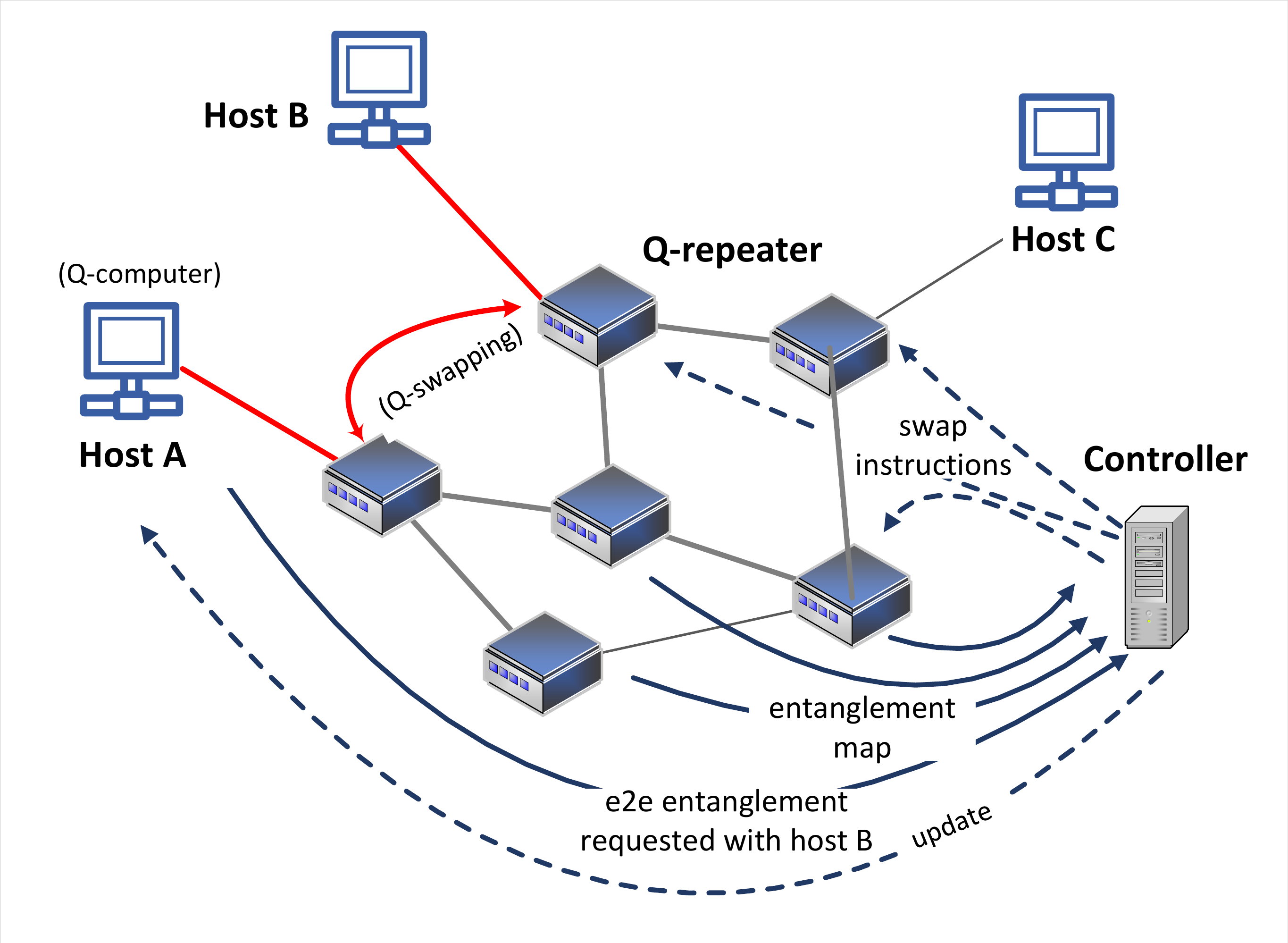}}
\caption{Quantum Network Architecture.}
\label{fig:net-architecture}
\end{figure}

\subsection{Topology, Basic Notation, and Fundamental Concept} \label{sec:topology}

We adopt a general model representing a quantum network with a set $\mathcal{N}$ of quantum nodes (i.e., quantum computers and repeaters). Each quantum node $u \in \mathcal{N}$ has a finite number of qubits $\mathcal{C}_u$ (i.e., memory and communication qubits).
The network traffic consists of a set $\mathcal{D}$ of demands. In particular, each demand reflects a source-destination pair. We denote $\mathcal{L}$ as the set containing links connecting quantum nodes. The links in the initial state reflect physical connections using fiber cables between quantum nodes as the quantum network that is built on top of the underlying conventional network system.

\begin{figure}[t]
\center
\subfigure{\includegraphics[width=8cm]{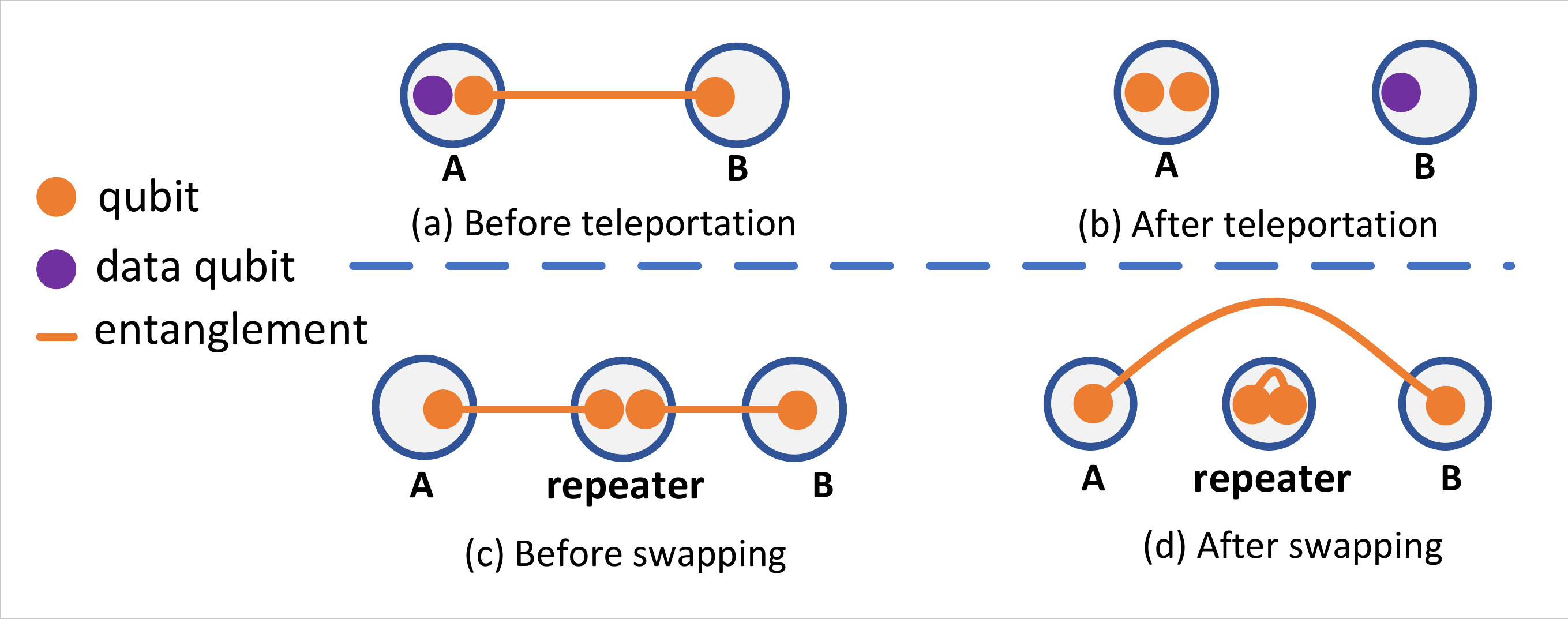}}
\caption{Transmission of qubit using teleportation in (a-b) and entanglement swapping in (c-d).}
\label{fig:tele-swap}
\end{figure}

\noindent\textit{1) Quantum entanglement and teleportation.} It is worth mentioning quantum entanglement as the heart of the quantum teleportation. Quantum entanglement offers strong connections between two or more particles (e.g., two photons) based on their correlation. In other words, quantum states of the two entangled particles remain linked to each other. Hence, if a pair of entangled qubits is distributed to two nodes,
one node can \textit{teleport} one qubit of information to the other using the Bell-state-measurement (BSM) \cite{Lutkenhaus1999} and a classical communication channel, this process is known as \text{quantum teleportation}  \cite{Bouwmeester1997}.
The states of two particles before and after the teleportation process are illustrated in Fig. \ref{fig:tele-swap}(a) and \ref{fig:tele-swap}(b), respectively.

\noindent\textit{2) Entanglement swapping.} Quantum repeaters relay entanglements to enable long-distance quantum communication using entanglement swapping \cite{Munro2015}, similar to the store-and-forward process in classical networks. A quantum repeater that holds the entanglements of two nodes can execute swapping to turn an \textit{indirect entanglement} to a \textit{direct entanglement}. Fig. \ref{fig:tele-swap}(c) and \ref{fig:tele-swap}(d) illustrate the entanglement swapping procedure.

\noindent\textit{3) $k$-entangled routing.} Key distribution-based applications in quantum networks require multiple-entangled routing paths due to the accelerated decay of entanglement. A pair of source and destination are said to be enabled with $k$-entangled routing if
there are at least $k$ paths\footnote{Paths are always disjoint because each entangled link always belongs to only a single path.} connecting them.

\subsection{$k$-Entangled Routing Problem and Illustration} \label{sec:problem}

In this section, we outline and mathematically formulate the $k$-entangled routing problem, which will lay the foundation for resilient entanglement routing in quantum networks.

\noindent\textbf{INSTANCE:} Suppose we are given a quantum network with a set $\mathcal{N}$ of quantum nodes in which each node $u \in \mathcal{N}$ has a finite number of qubits $\mathcal{C}_u$, the set of incoming demands $\mathcal{D}$, and a set $\mathcal{L}$ containing links connecting quantum nodes.

\noindent\textbf{QUESTION:} Does an entanglement routing schedule ($\mathcal{RS}$) for all demands (source-destination pairs) exist in the network, such that the number of entangled paths between any source-destination pair is not less than $k$?

In particular, if there are more dynamic routing options between a source-destination pair, the flexibility of the network can be further increased.
Since qubits inside of a quantum node can entangle to qubits in other nodes and they can also swap with the other qubits located in the same node in different orders, the data flow of a source-destination pair may traverse through differently entangled alternate paths. The \textit{minimum traffic flexibility} of all demands in the network can be mathematically expressed as follows:\vspace{-5pt}

\begin{equation}\label{eq:k}
k=\underset{(s,d)\in \mathcal{D}}{\mathop{\min }}\,\sum\limits_{_{(s,d)}}{{{\mathcal{P}}_{(s,d)}}{{1}_{\left\{ {{\mathcal{E}}_{(u,v)}}=1\forall (u,v)\in {{\mathcal{P}}_{(s,d)}} \right\}}}}
\end{equation}

\noindent
Where $\mathcal{P}_{(s,d)}$ denotes the entangled path connecting the source-destination pair $(s,d)$.
Here, $1_{\{.\}}$ is the indicator function; it is equal to one if the condition in the subscript is true, otherwise zero.
The natural objective of the $\mathcal{RS}$ is, therefore, to maximize the \textit{traffic flexibility} for all demands by increasing the number of dynamically entangled paths by leveraging quantum entanglement and swapping.\vspace{-5pt}

\begin{equation}\label{eq:obj}
\textbf{Obj:}\text{  } \underset{(s,d)\in \mathcal{D}}{\mathop{\operatorname{maximize}}}\,\text{  }k
\vspace{-2pt}
\end{equation}

\begin{equation}\label{eq:const1}
\text{s.t.  } constraints: (\ref{eq:k})
\vspace{-2pt}
\end{equation}

\begin{equation}\label{eq:const2}
\mathcal{E}_{(u,v)} \in \{0,1\}, \forall\ (u,v)\in \mathcal{P}_{(s,d)}
\end{equation}

Where $\mathcal{E}_{(u,v)}$ denotes the entanglement status between the two quantum nodes $u$ and $v$. Constraint \ref{eq:const2} indicates that any link on any feasible path connecting the source and destination must be an entangled link.
The objective function is a challenging combinatorial optimization problem. The optimization functions potentially come with an exponential number of variables that cannot be solved in polynomial time in general. The number of possible entanglements and swappings rise exponentially with the network scale$-$\textit{the combinatorial curse}.

The investigated problem can be expressed as the general case (with multiple source-destination pairs) of the node-disjoint path (NDP) problem on graphs that aims also to find a set of paths linking a specified pair of nodes such that no two paths share a node. The hardness of the $k$-entangled routing problem is easily indicated because the NDP problem is shown as NP-complete \cite{vygen1995np}. Due to space limitation, we do not elaborate the detailed proof here. In the following sections, we present two different approaches to optimize this objective.
Particularly, in section \refsec{smpsa} we propose the Sequential Multi-Path Scheduling Algorithm (SMPSA) and argue that it is possible to first construct at least $1$-entangled routing for all demands based on the initial network topology and then the $k$-entangled routing will be leveraged by utilizing the available qubits and repeaters.
However, alone, it fails to achieve the $k$-entangled routing for all demands if the availability of network resources is very tight.
This leads us to introduce the second algorithm, Min-Cut-based Multi-Path Scheduling Algorithm (MCSA) in section \refsec{mcsa}, which considers to achieve the $k$-entangled routing for all demands at the same time.

\section{Optimizing the traffic flexibility} \label{sec:algos}

In this section we aim to optimize the traffic flexibility by considering the $k$-entangled routing problem, wherein each demand $(s_i,d_i)$ is considered to be connected by at least $k$ entangled paths\footnote{Paths are "qubit-disjoint"}. We then analyze preliminary properties of the proposed algorithms in \refsec{fact}, and discuss the convoluted nature of $k$-entangled routing via experiments, in \refsec{netsquid}.

\subsection{Sequential Multi-path Scheduling Algorithm (SMPSA)} \label{sec:smpsa}

\begin{algorithm} [http]
\caption{Sequential Multi-path Scheduling Algorithm} \label{alg:smpsa}
\textbf{Input:} $\mathcal{N}, \mathcal{L}, \mathcal{C}_u, \mathcal{D}$\\
\textbf{Output:} $k$
\begin{algorithmic}[1]

\State Construct $\mathcal{G}_e = (\mathcal{N}_e, \mathcal{L}_e)$

\State Let $\mathcal{P}_{(src, dst)}$ be the set of entangled paths allocated for $(src, dst)$ $: \mathcal{P}_{(src,dst)} \leftarrow \emptyset$ $\forall\ (src, dst) \in \mathcal{D}$

\While{$\mathcal{D} \ne \emptyset$}

    \State Let $(src, dst)$ be the first demand in the ordered list of demands $\mathcal{D}$

    \State $\mathcal{D} \leftarrow \mathcal{D} - \{(src,dst)\}$

    \State Let $p_{(src, dst)}$ be the shortest path from $src$ to $dst$ in $\mathcal{G}_e$ such that any two adjacent nodes on the shortest path are entangled to each other.

    \If{$\exists$ $p_{(src, dst)}$}

        \State $\mathcal{P}_{(src, dst)} \leftarrow \mathcal{P}_{(src, dst)} + \{p_{(src, dst)}\}$

        \State $\mathcal{L}_e \leftarrow \mathcal{L}_e - \{(u,v)\}$ $\forall\ (u,v) \in p$

        \State $\mathcal{D} \leftarrow \mathcal{D} + \{(src,dst)\}$

    \EndIf

\EndWhile

    \State return $k\leftarrow \underset{\left( src,dst \right)\in \mathcal{D}}{\mathop{\min }}\,\left| {{\mathcal{P}}_{\left( src,dst \right)}} \right|$

\end{algorithmic}
\end{algorithm}

We consider the following designs in the post-steps of the initialization of an \textit{entangled quantum network} in which all entangled links are already determined and we denote the entangled network by $\mathcal{G}_e = (\mathcal{N}_e, \mathcal{L}_e)$.
Given $\mathcal{G}_e = (\mathcal{N}_e, \mathcal{L}_e)$, the problem of demand scheduling and path allocation poses several unique complexities.
In this section, we present the sequential multi-path scheduling algorithm (SMPSA) to tackle the $k$-entangled routing problem (eq. \ref{eq:obj}).
The key idea behind the design of the SMPSA is to first construct at least $1$-entangled routing for all demands based on the initial network topology and then the $k$-entangled routing will be enabled by utilizing available qubits and repeaters.
We note that the SMPSA is designed as a centralized algorithm, and therefore each node has up to date global link-state knowledge.
In particular, the SMPSA maintains the ordered set of demands $\mathcal{D}$ using a first-come first-serve (FCFS) policy (thereby maintaining fairness), and accommodates the available paths for a given demand sequentially i.e. one at a time.
On the basis thereof, let $(src,dst)$ be the first demand in $\mathcal{D}$. We denote $p_{(src,dst)}$ and $\mathcal{P}_{(src,dst)}$ as the shortest path and the set of shortest paths connecting $src$ and $dst$, respectively. In each iteration, $(src, dst)$ is removed from $\mathcal{D}$, and accommodated with the shortest path $p_{(src, dst)}$ on the basis of the availability of network resources i.e. qubits and quantum links. If the path $p_{(src, dst)}$ exists, the SMPSA allocates the resources for $p_{(src, dst)}$ and pushes $(src, dst)$ back into the rear of $\mathcal{D}$. However if the network cannot support $p_{(src, dst)}$, the demand $(src, dst)$ is accommodated within the next time slot based on the number of paths allocated $|\mathcal{P}_{(src, dst)}|$.
We remark that the re-addition of $(src, dst)$ into $\mathcal{D}$ serves as the pivotal step towards ensuring fair path allocation for the SMPSA. This is primarily because the sequential allocation of paths provides a fair-share of network resources for each demand, in contrast to the increased chances of resource-contention between the demands in $\mathcal{D}$ for the allocation of multiple paths on each iteration. However, we observe that the FCFS queuing policy of the demands prioritizes older demands.
Upon the allocation of $p_{(src, dst)}$, the capacities and quantum links within $\mathcal{G}_e$ are updated appropriately, and the reduced subgraph $\mathcal{G}_e'$ is considered for the future iterations such that $\mathcal{G}_e \leftarrow \mathcal{G}_e'$. Finally, the aforementioned procedure is repeated until there are no available paths for all $(src, dst)$ in $\mathcal{D}$, after which the value of $k$ is determined and returned. Algorithm \ref{alg:smpsa} describes the working of the SMPSA in detail, and the time complexity of the SMPSA is derived within Theorem \ref{thm:time-smpsa}. Although the SMPSA obtains good results in an efficient manner, the prioritization of resources towards older demands poses performance issues towards obtaining the highest value of $k$ for all demands in the network. In the following section, we discuss an alternative paradigm i.e. the MCSA, which aims to provide a better performance using the concept of the min-cut for the prioritization of demands, along with the concurrent allocation of paths while guaranteeing a fair time complexity.

\subsection{Min-Cut-based Multi-path Scheduling Algorithm (MCSA)} \label{sec:mcsa}

\begin{algorithm} [http]
\caption{Min-Cut-based Multi-path Scheduling Algorithm} \label{alg:mcsa}
\textbf{Input:} $\mathcal{N}, \mathcal{L}, \mathcal{C}_u, \mathcal{D}$\\
\textbf{Output:} $k$
\begin{algorithmic}[1]

\State Construct $\mathcal{G}_e = (\mathcal{N}_e, \mathcal{L}_e)$

\State Let $\mathcal{P}_{(src, dst)}$ be the set of entangled paths allocated for $(src, dst)$ $: \mathcal{P}_{(src,dst)} \leftarrow \emptyset$ $\forall\ (src, dst) \in \mathcal{D}$

\While{$\mathcal{D} \ne \emptyset$}

    \For{$(s,d) \in$ $\mathcal{D}$}

        \State Let $C_{(s,d)}$ be the minimum cut that separates $s$ and $d$ from each other in $\mathcal{G}_e$.

    \EndFor

    \State Let $(src, dst)$ be the demand in $\mathcal{D}$ having the minimum path flexibility $\mathcal{F}(src, dst)$

    \State $\mathcal{D} \leftarrow \mathcal{D} - \{(src,dst)\}$

    \State Let $n$ represent the maximum number of paths allocated for a demand $:n_{(src, dst)} \leftarrow \min(\mathcal{C}_{src}, \mathcal{C}_{dst})$

    \While{$n \ne 0$}

        \State Let $p_{(src, dst)}$ be the shortest path from $src$ to $dst$ in $\mathcal{G}_e$ such that any two adjacent nodes on the shortest path are entangled to each other.

        \If{$p_{(src, dst)} \ne \emptyset$}

            \State $\mathcal{P}_{(src, dst)} \leftarrow \mathcal{P}_{(src, dst)} + \{p_{(src, dst)}\}$

            \State $\mathcal{L}_e \leftarrow \mathcal{L}_e - \{(u,v)\}$ $\forall\ (u,v) \in p_{(src, dst)}$

            \State $n_{(src, dst)} \leftarrow n_{(src, dst)} - 1$

        \Else

            \State $n_{(src, dst)} \leftarrow 0$

        \EndIf

    \EndWhile

\EndWhile

\State return $k\leftarrow \underset{\left( src,dst \right)\in \mathcal{D}}{\mathop{\min }}\,\left| {{\mathcal{P}}_{\left( src,dst \right)}} \right|$

\end{algorithmic}
\end{algorithm}

The $k$-entangled routing obtained through the SMPSA in \refsec{smpsa} attains a fair value $k$ for entangled routing paths
of all demands we desire. However, it may not be sufficient to maximize $k$ through the prioritization of fairness, since the resource availability of the demands are not considered.
In order to attain the maximum value of $k$, in this section we propose the Min-Cut-based Multi-path Scheduling Algorithm (MCSA) that performs well within reasonable bounds of computational complexity.
The core idea of the MCSA is the prioritization of resources towards the demands with lesser path flexibility.
Given an entangled quantum network $\mathcal{G}_e = (\mathcal{N}_e, \mathcal{L}_e)$ with unit edge-weights $\omega(l):\omega(l) = 1$ $\forall\ l \in \mathcal{L}_e$, we define the path flexibility of a demand (i.e., $(s,d)$) $\mathcal{F}((s, d))$ as the minimum number of edges required to disconnect any entangled path from $s$ to $d$ such that $(s, d) \in \mathcal{D}$, and mathematically formulate it as follows:

\begin{equation}\label{eq:path-flex}
    \mathcal{F}(s, d) = |C_{(s, d)}|\ \forall\ (s, d) \in \mathcal{D}
\end{equation}

In eq. \ref{eq:path-flex}, $C_{(s, d)}$ denotes the minimum $s$-$d$ cut set, and the path flexibility $\mathcal{F}(s,d)$ is given by the cardinality of $C_{(s, d)}$. We note that since $\mathcal{G}_e$ consists of unit edge-weights the capacity of $C_{(s, d)}$ i.e. $\sum\limits_{e \in C_{(s, d)}}\sum\limits_{(u, v) \in e} \omega(u,v)$, is equivalent to the cardinality of $C_{(s, d)}$ i.e. $\mathcal{F}(s, d)$.
For each iteration of the MCSA, we firstly compute the minimum $s$-$t$ cut of all the demands $(s, d) \in \mathcal{D}$, after which the demand $(src, dst)$ with the minimum path flexibility $\mathcal{F}(src, dst)$ is chosen. To accommodate multiple paths for $(src, dst)$ concurrently, we compute the value of $n_{(src, dst)}$ to determine the maximum number of paths which can be supported from $src$ to $dst$ through the minimum capacity of the $src$ and $dst$ i.e. $\min(\mathcal{C}_{src},\mathcal{C}_{dst})$. Upon the computation of $n_{(src, dst)}$, the shortest path $p_{(src, dst)}$ is allocated iteratively so as to ensure the absence of resource contention on the allocation of each path. Once a given path $p_{(src, dst)}$ is allocated, the links and capacities within the current network $\mathcal{G}_e$ is updated to $\mathcal{G}_e'$ such that $\mathcal{G}_e \leftarrow \mathcal{G}_e'$. Unlike the SMPSA, the above process is repeated until $\mathcal{D}$ is empty, and hence $k$ is obtained. The detailed operation of the MCSA is described in Algorithm \ref{alg:mcsa}, and its time complexity is derived within Theorem \ref{thm:time-mcsa}. In addition, we provide thorough performance analysis of the proposed algorithms in \refsec{eval}.

\subsection{Preliminary Facts} \label{sec:fact}

\begin{thm} \label{thm:time-smpsa}
The time complexity of the SMPSA is bounded in $O(|\mathcal{C}| \times |\mathcal{N}|)^3 \times |\mathcal{D}|)$.
\end{thm}

\begin{proof}
Because the are at most $|\mathcal{C}| \times |\mathcal{N}|$ qubits and $\frac{|\mathcal{C}| \times |\mathcal{N}| (|\mathcal{C}| \times |\mathcal{N}| - 1)}{2}$ entangled links, the construction of $\mathcal{G}_e = (\mathcal{N}_e, \mathcal{L}_e)$ requires $O((|\mathcal{C}| \times |\mathcal{N}|)^2)$. In the while loop of the SMPSA, the sigle shortest path for a pair of source and destination can be determined within $O(|\mathcal{C}| \times |\mathcal{N}| \times \frac{|\mathcal{C}| \times |\mathcal{N}| (|\mathcal{C}| \times |\mathcal{N}| - 1)}{2})$ using the Bellman-Ford algorithm \cite{7916533}. It is bounded in $O((|\mathcal{C}| \times |\mathcal{N}|)^3)$. There are $|\mathcal{D}|$ number of demands in the network, then it requires at most $O((|\mathcal{C}| \times |\mathcal{N}|)^3 \times |\mathcal{D}|)$ to determine the shortest path for all demands.
Therefore, the time complexity of the SMPSA is bounded in $O((|\mathcal{C}| \times |\mathcal{N}|)^2 + (|\mathcal{C}| \times |\mathcal{N}|)^3 \times |\mathcal{D}|) = O(|\mathcal{C}| \times |\mathcal{N}|)^3 \times |\mathcal{D}|)$.
\end{proof}

\begin{thm} \label{thm:time-mcsa}
The time complexity of the MCSA is bounded in $O(|\mathcal{C}|^2 \times |\mathcal{N}|)^5 \times |\mathcal{D}|^2)$.
\end{thm}

\begin{proof}
As indicated in above analysis, the construction of $\mathcal{G}_e = (\mathcal{N}_e, \mathcal{L}_e)$ requires $O((|\mathcal{C}| \times |\mathcal{N}|)^2)$. In the outer while loop, it requires $O((|\mathcal{C}| \times |\mathcal{N}|)^2)$ to return the minimum cut for a demand in the network using the Karger’s algorithm \cite{10.1145/234533.234534}. There are $|\mathcal{D}|$ number of demands in the network, then it requires at most $O((|\mathcal{C}| \times |\mathcal{N}|)^2 \times |\mathcal{D}|)$ to determine the minimum cut for all demands. It also requires $O((|\mathcal{C}| \times |\mathcal{N}|)$ to compute the  maximum number of paths allocated for a demand.
In the inner while loop, it requires at most $O((|\mathcal{C}| \times |\mathcal{N}|)^3 \times |\mathcal{D}|)$ to determine the shortest path for all demands. We therefore have that the complexity of the MCSA is bounded in $O((|\mathcal{C}| \times |\mathcal{N}|)^2 \times |\mathcal{D}| \times (|\mathcal{C}| \times |\mathcal{N}|)^3 \times |\mathcal{D}|) = O(|\mathcal{C}|^2 \times |\mathcal{N}|)^5 \times |\mathcal{D}|^2)$, which thus completes the proof.

\end{proof}

To demonstrate the feasibility of the proposed solutions, We firstly consider the $k$-entangled routing problem in different network settings, wherein the network topology generated after the entanglement $\mathcal{G}_e$ is the grid. Let $B$ and $B'$ be the upper and lower boundary of the network, respectively. In particular, $B$ and $B'$ are the two rows located in the upper and lower boundary of the network. We show that the following observation always holds true.

\begin{thm}
Given the quantum network in the form of grid, $\mathcal{G}_e = (r \times c)$-grid. Then for any pair of source and destination $(src,dst)$, in which $src \in V(B)$ and $dst \in V(B')$, there always exists the set $\mathcal{P}_{(src,dst)}$ of entangled paths connecting the $src$ and $dst$. In addition, $\mathcal{P}_{(src,dst)}$ can be determined efficiently in polynomial time.
\end{thm}

\begin{proof}
Let $\mathcal{G}'_e$ be the graph that induced by removing all nodes of $(B \bigcup B') \backslash \mathcal{D}$. The property is equivalent with that if there exists the set of entangled paths $\mathcal{P}_{(src,dst)}$ connecting $src$ and $dst$ in $\mathcal{G}'_e$ for all $(src,dst) \in \mathcal{D}$.
According to the Menger's theorem \cite{bohme2001menger}, the size of a minimum cut set is equal to the maximum number of disjoint paths that can be found between any pair of vertices in the graph. It means that for $\mathcal{G}'_e \backslash C$, in which $C$ is the minimum cut of the network, there is no path connecting any pair of source and destination in $\mathcal{D}$. In fact, the source and destination nodes in $\mathcal{D}$ may not contain a node in $C$. Some columns, say $M$, do not contain a node in $C$, and some others $M'$ contain. Since, there are at least $|\mathcal{D}| + 2$ rows in $G$, there are some rows, say $R$, that do not contain any node in $C$. We then have that $(M \bigcap R \bigcap M') \bigcup \mathcal{G}'_e$ belong to the same connected component of $\mathcal{G}'_e \backslash C$ and this connected component contains the demand in $\mathcal{D}$, a contradiction. The set of entangled paths $\mathcal{P}_{(src,dst)}$ therefore can be obtained efficiently by the maximum single-commodity flow between the source and destination.
\end{proof}

\subsection{From Quantum Entanglement and Teleportation to Routing}\label{sec:netsquid}

Quantum entanglement and teleportation have evolved from communications of two quantum nodes, where data is delivered \textit{point-to-point} from a source to a \textit{neighbor} destination, to communications of a long distant source-destination pair with the help of quantum repeaters. Such an architecture is employed not only for a small network but also for a large-scale networks. In either case, timely delivery of routing information is crucial, as it can have a significant impact on the success of quantum teleportation and thus the bottom-line of service providers and quantum network architecture design. In this section, we implement a realistic quantum network using a software tool, called ``NetSquid" \cite{Coopmans2021} for modelling the stochastic processes underlying quantum communication. The parameters used to setup the quantum network are aptly tabulated, and the results organized within this section are incorporated within \refsec{eval} to better understand how the proposed algorithms would fit into the operation of practical large scale quantum networks.\\

\begin{tabular}{ |p{3.8cm}||p{3.5cm}|}
 \hline
 \multicolumn{2}{|c|}{Quantum Network Parameters} \\
 \hline
 Parameter&Value\\
 \hline
 Number of nodes ($|\mathcal{N}|$)    &   $10$ \\
 Propagation speed ($c$)   &   \SI{200000}{\km\per\second} \\
 Length of link ($l_d$) & (\num{1}, \num{2.5}, \num{5}, \num{7.5}) \si{\km}\\
 Propagation delay ($t_{l}$)   &   $(l_d/c)$ \si{\second}\\
 Source frequency ($f_{s})$   &   \num{21.05} $l_d$ \si{\hertz}\\
 Dephasing rate ($r_{deph}$)   &   [$\num{1}$, $\num{10}^4$] \si{\mega\hertz}\\
 Depolarization rate ($r_{depo}$)   &   [$\num{10}^{-3}$, $\num{10}$] \si{\mega\hertz}\\
 Fidelity ($F$) &   [$0.0, 1.0$]\\
 \hline
\end{tabular}
\vspace{5pt}

\noindent \textbf{Network setup.} We conduct experiments on a quantum network with the following parameters: the number of nodes $|\mathcal{N}|$ such that $|\mathcal{N}|=10$;
the propagation speed ($c$) which indicates the speed of light in a fiber optic cable i.e. $c= \SI{200000}{\km\per\second}$;
the length of the fiber optic cable denoted by $l_d$ is varied throughout the experiments such that $l_d = (\num{1}, \num{2.5}, \num{5}, \num{7.5})\ \si{\km}$; the propagation delay ($t_l$) indicates the delay in seconds (\si{\second}) for a qubit to be transmitted through its respective fiber optic link, and is given as $t_l = \frac{l_d}{c}$ \si{\second}; the source frequency ($f_s$) indicates the frequency at which the source nodes create entangled qubits i.e. $f_s$ = \num{21.05} $l_d$ \si{\hertz};
the dephasing rate ($r_{deph}$) indicates the exponential depolarizing rate used to apply dephasing noise to qubit(s) on a quantum component, and its value is varied from $\num{1}$ $\si{\mega\hertz}$ to $\num{10}^4$ $\si{\mega\hertz}$ within the experiments; similarly, the depolarization rate ($r_{depo}$) indicates the the exponential depolarizing rate per unit time ($Hz$), however it is used to apply depolarizing noise to qubit(s) on a quantum component, and is varied from $\num{10}^{-3}$ $\si{\mega\hertz}$ to $\num{10}$ $\si{\mega\hertz}$ during the experiment; Finally, the fidelity ($F$) is obtained via the simulations to show the closeness of two quantum states, and its value ranges from $0.0$ to $1.0$.

\noindent \textbf{Properties of fidelity.} Experimental studies \cite{Jozsa94} have demonstrated that fidelity reflects the accuracy of the transmission of a qubit. In other words, ``it expresses the probability that one state will pass a test to identify as the other" \cite{Jozsa94} and primarily occurs due to the following reasons:

\begin{itemize}
    \item The data being compressed before transmission, and;
    \item The communication channel (e.g. fiber optic cable) being subject to random noise.
\end{itemize}

During the transmission of pure quantum states, the states are generally degraded slightly resulting in the receiver obtaining mixed quantum states. Since, fidelity quantifies how well entanglement is maintained and persevered through quantum channels \cite{Mousolou2020}, it plays an essential role for the realization of quantum networks. Given two pure quantum states $\rho$ and $\sigma$, the fidelity function is obtained as follows:\vspace{-10pt}

\begin{equation}\label{eq:fidelity}
F(\rho, \sigma)= \left\{\mathrm{trace}\sqrt{\sqrt{\rho}\sigma\sqrt{\rho}}\right\}^2
\end{equation}

\noindent \textbf{Results. }To model the convoluted processes underlying pragmatic quantum networks, we conduct two experiments (Fig. \ref{fig:phase3-results}) to better understand the effect of the dephasing noise rate ($r_{deph}$), depolarization noise rate ($r_{depo}$), and the length the quantum links ($l_d$), on the quantum fidelity ($F$). Specifically, in Fig. \ref{fig:dephase} and \ref{fig:depolar} we analyze the average value of $F$ for $2000$ experiments, on increasing $r_{deph}$ and $r_{depo}$ respectively, over unique values of $l_d$. Clearly, from the experiments conducted in Fig. \ref{fig:phase3-results}, $r_{deph}, r_{depo}$, and $l_d$ considerably effect the value of $F$. Furthermore, since the value of $F$ has a large effect on the connectivity of the network, we introduce the overall entanglement generation probability ($p_s$) in \refsec{eval} to pragmatically model the performance of the proposed paradigms.

\begin{figure}[t]
\center
\subfigure[]{\includegraphics[width=8cm]{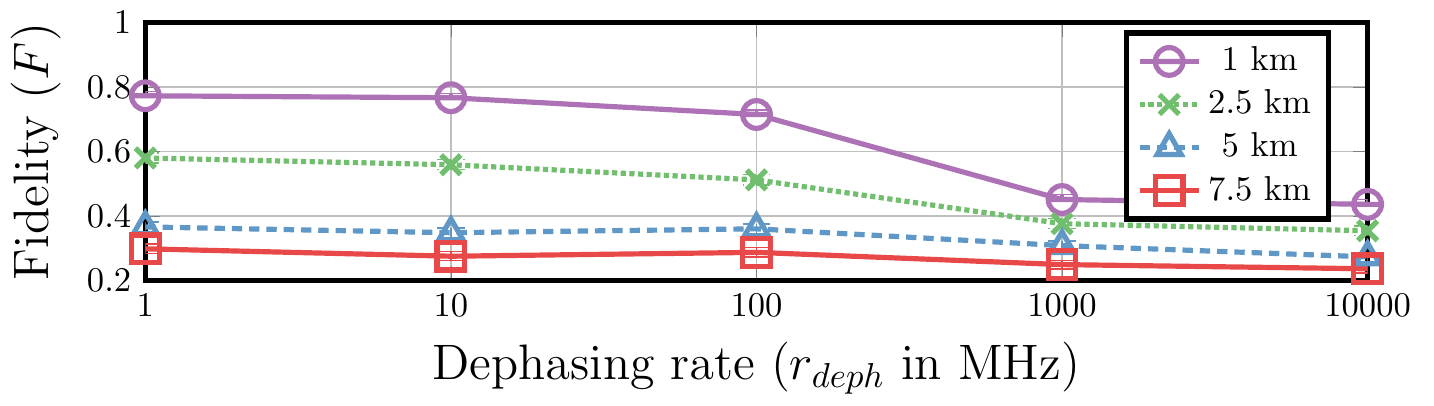} \label{fig:dephase}}
\subfigure[]{\includegraphics[width=8cm]{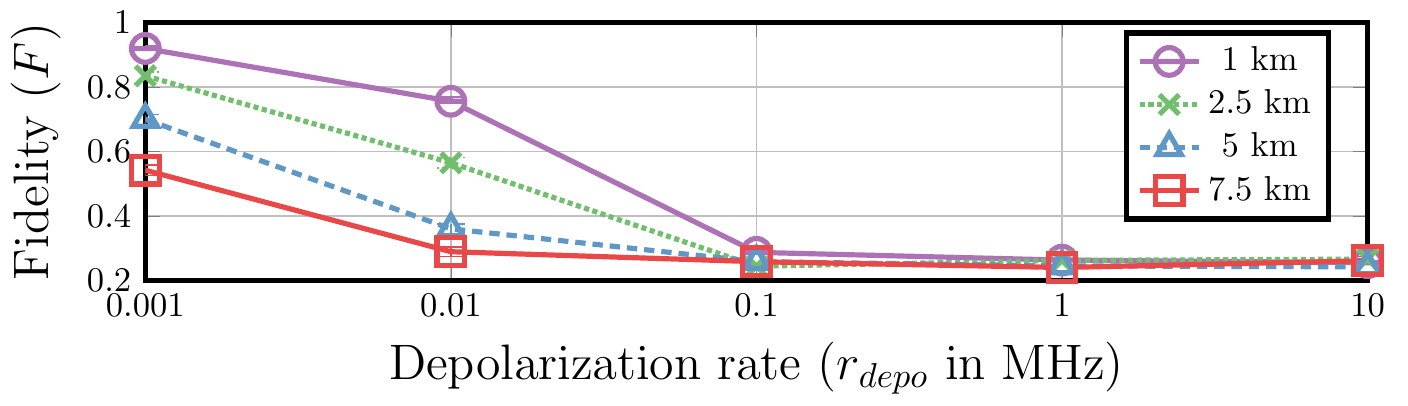} \label{fig:depolar}}
\vspace{-5pt}
\caption{Fidelity of quantum states when a) the dephasing rate ($r_{deph}$) ranging from $\num{1}$ to $\num{10}^4$ \si{\mega\hertz}, and b) the depolarization rate ($r_{depo}$) ranging from $10^{-3}$ to $10$ \si{\mega\hertz}.}
\vspace{-15pt}
\label{fig:phase3-results}
\end{figure}

\lstset{frame=tb,
  language=tcl,
  aboveskip=3mm,
  belowskip=3mm,
  showstringspaces=false,
  columns=flexible,
  basicstyle={\small\ttfamily},
  numbers=none,
  breaklines=true,
  breakatwhitespace=true,
  tabsize=3
}

\begin{lstlisting}[caption={Dephase log file},captionpos=b, label={lst:log}]
0.228 secs: Setting node distance of 1.0 km.
------------------------------------------
0.228 secs: Setting the dephase noise rate as 1000000.00 Hz.
------------------------------------------
0.230 secs: Initializing 10 nodes in the network with 1.0 km between the nodes
0.230 secs: Creating EntanglingConnection between node 0 and 7.
...
0.238 secs: Calculating the fidelity.
0.238 secs: The fidelity is 0.773.
...
12.941 secs: Setting node distance of 2.5 km.
\end{lstlisting}

In the following we show the topology of the quantum network we used to conduct the experiments above.

\begin{figure}[http]
\center
\subfigure{\includegraphics[width=6cm]{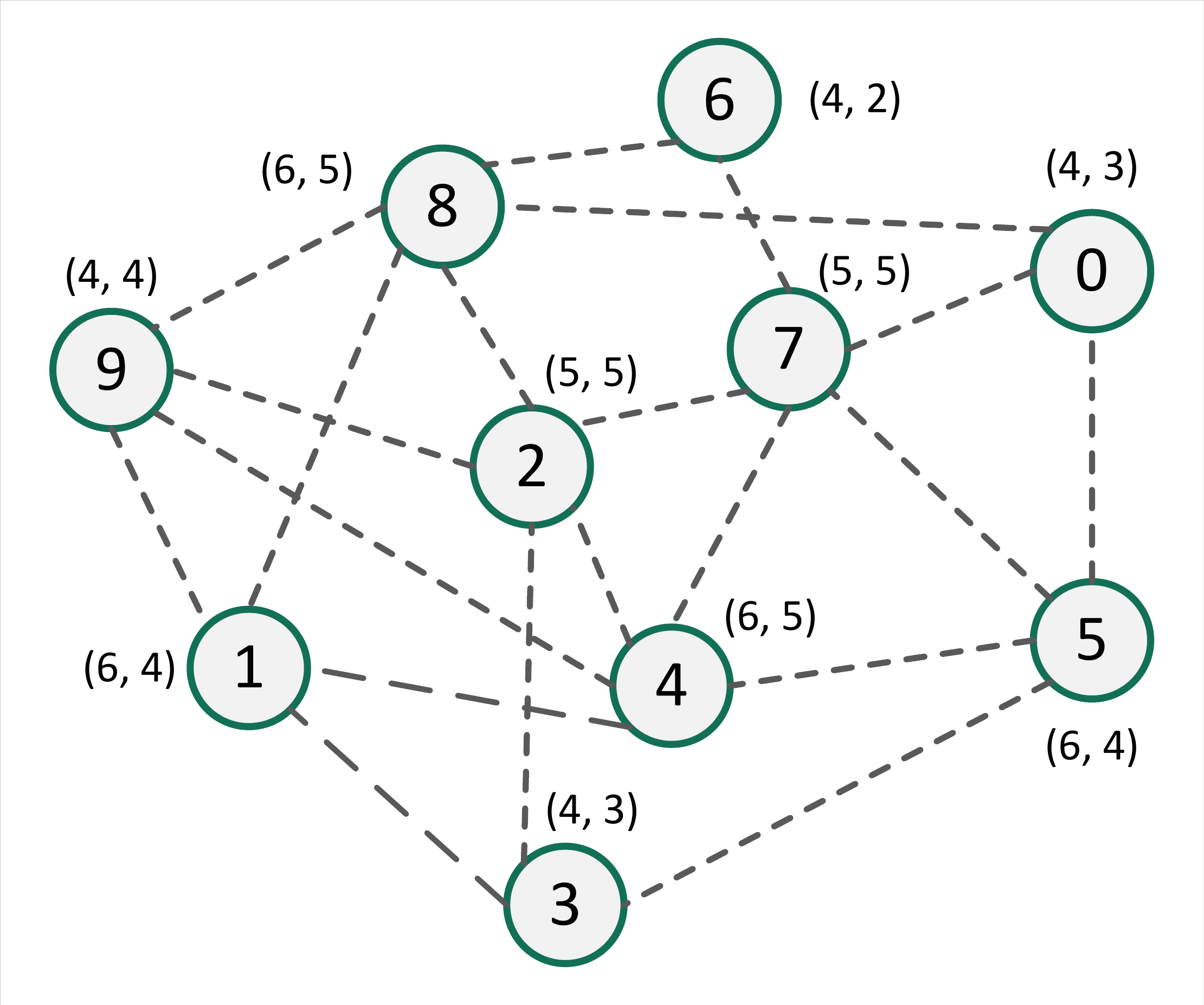}}
\caption{Network topology and qubit settings.}
\label{fig:phase3-topology}
\end{figure}

\section{Performance Evaluation}\label{sec:eval}

In this section, we assess the performance of the scheduling algorithms by implementing the proposed paradigms and performing extensive simulations. To better characterize the performance, we evaluate the algorithms on the basis of three primary metrics; (i) scheduling efficiency ($\refsec{sim-scheff}$), (ii) path selection efficiency ($\refsec{sim-patheff}$), and (iii) resource utilization efficiency  ($\refsec{sim-reseff}$). Furthermore, we define two additional scheduling paradigms, namely the Random Multi-path Scheduling Algorithm (RMPSA) and Distance Multi-path Scheduling Algorithm (DMPSA). The RMPSA maintains a first-come first-serve based queuing policy, and selects the paths for a given demand randomly. Similarly, the DMPSA also maintains a first-come first-serve based queuing policy, however the path selection is on the basis of the minimum physical distance for a given demand.

We setup the quantum network in the simulation environment consisting of $50,100, 150, 200,$ or $250$ nodes, and a randomized topology to fully reflect the complexity of the problem. Furthermore, each link within the network is allotted a physical distance ($d_l:l\in\mathcal{L}$) such that the average physical distance of the links in the network is $7.44$ km, $12.87$ km, $17.49$ km, $22.48$ km, or $27.5$ km. The physical distance metric supports the simulation of the overall entanglement generation probability $p_s$ such that $p_s = e^{-\alpha d_l}$, where $\alpha$ is a constant depended on the transmission media (e.g. fiber optic link). Lastly, the average capacity ($\mathcal{C}_{avg})$ of the quantum nodes is set to $3, 5.05, 6.95, 9.09,$ or $10.96$, and the number of demands ($|\mathcal{D}|$) within the ordered set is taken as $5, 10, 15, 20$ or $25$.

\subsection{Scheduling Efficiency}\label{sec:sim-scheff}

\begin{figure}[http]
\center
\subfigure[]{\includegraphics[width=4.2cm]{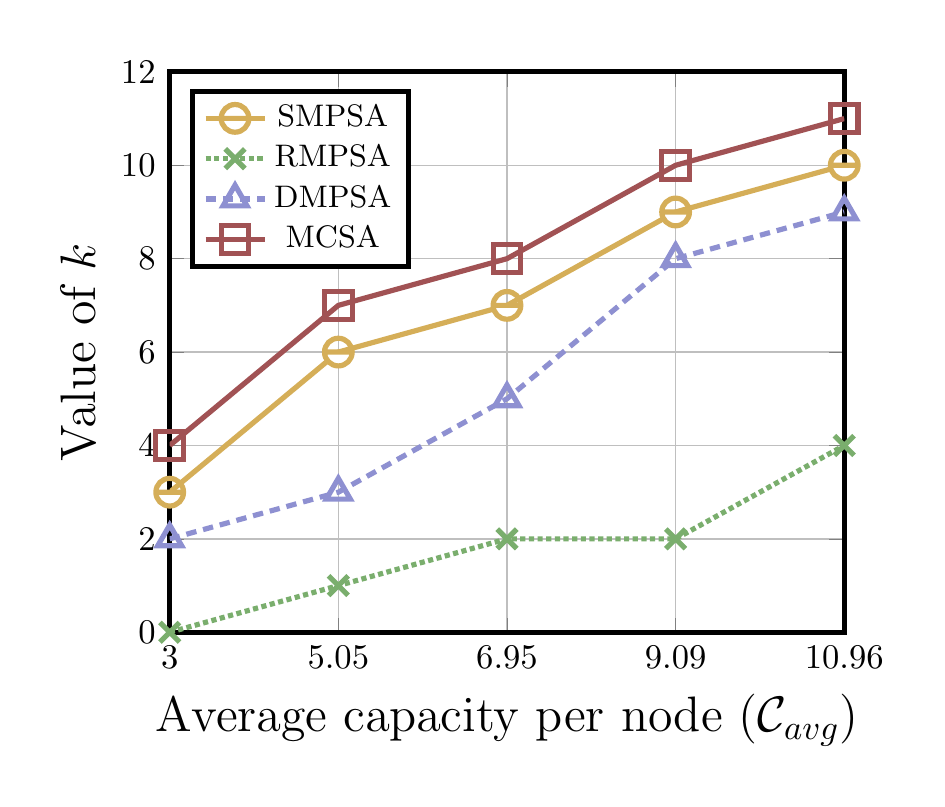} \label{fig:scheff_avg_cap}}
\subfigure[]{\includegraphics[width=4.2cm]{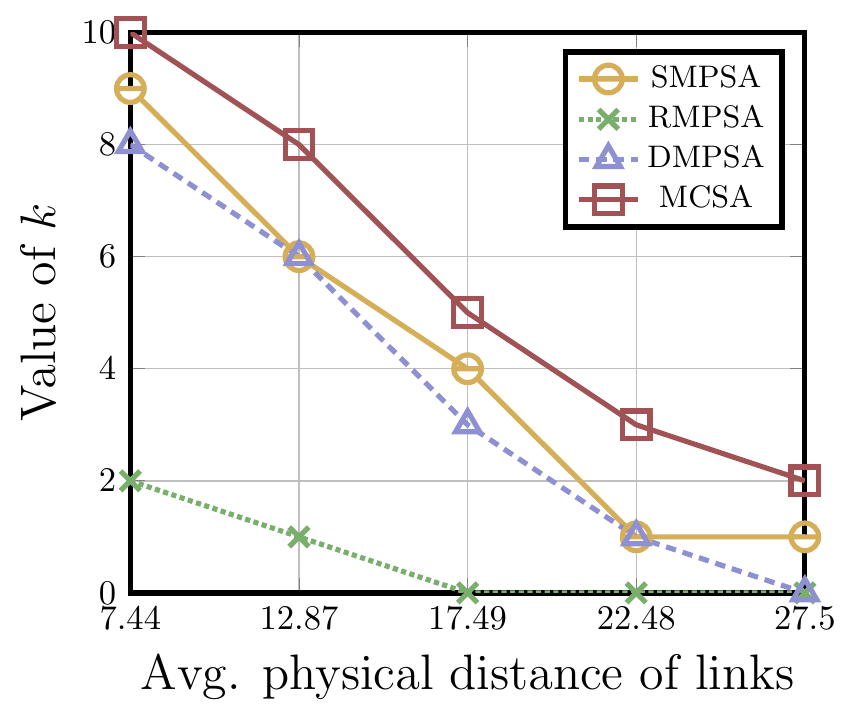}\label{fig:scheff_avg_phys}}
\subfigure[]{\includegraphics[width=4.2cm]{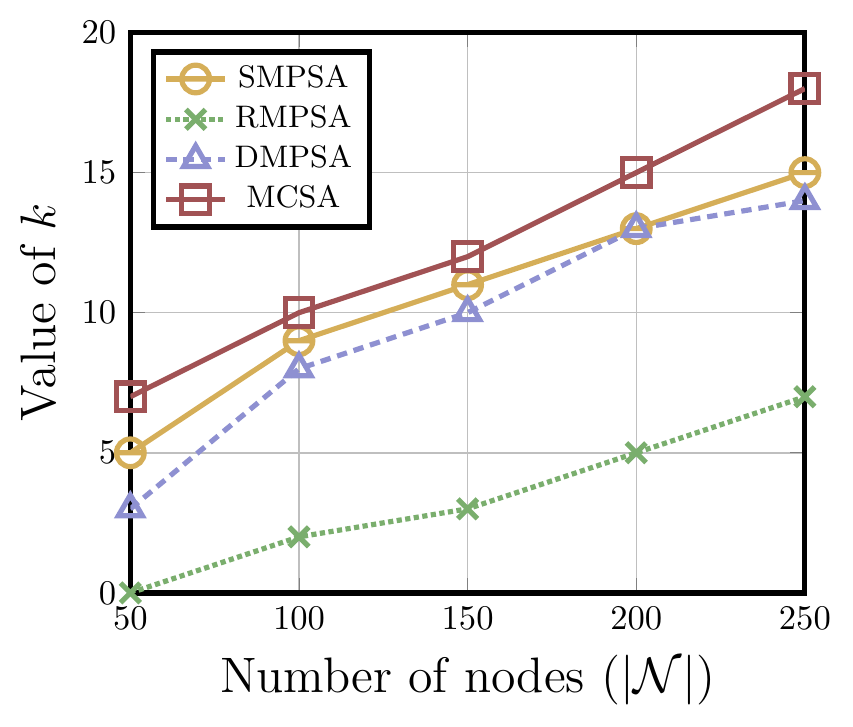} \label{fig:scheff_num_node}}
\subfigure[]{\includegraphics[width=4.2cm]{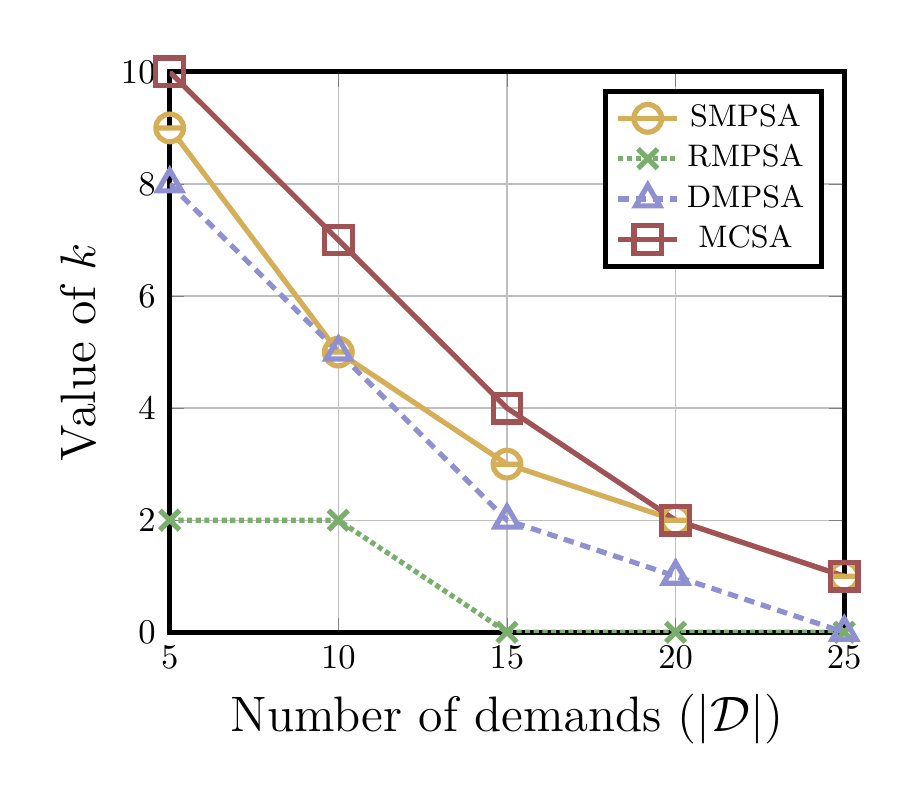}\label{fig:scheff_num_dem}}
\vspace{-5pt}
\caption{The value of $k$ when a) the average capacity per node ($\mathcal{C}_{avg}$) ranging from $3$ to $10.96$, b) average physical distance of links ranging from $7.44$ to $27.5$, c) number of nodes ($|\mathcal{N}|$) ranging from $50$ to $250$, and d) number of demands ($|\mathcal{D}|$) ranging from $5$ to $25$.}
\vspace{-10pt}
\label{fig1:ph1-kval}
\end{figure}

We first evaluate the scheduling efficiency of the algorithms by comparing the values of $k$ obtained by the SMPSA, MCSA, RMPSA, and DMPSA under five distinct network environments. Fig. \ref{fig1:ph1-kval} and Fig. \ref{fig:ph1-avgkval} demonstrate the scheduling performance of the algorithms in terms of the value of $k$, and the average value of $k$ for distinct simulation iterations respectively.

Fig. \ref{fig:scheff_avg_cap} shows the comparison of the value of $k$ when the average capacity per node ($\mathcal{C}_{avg}$) is increasing from $3$ to $10.96$. In this case, the network parameters are setup such that, the average physical distance is set to $7.44$, $|\mathcal{N}| = 100$, and $|\mathcal{D}| = 5$. As illustrated within the figure, the MCSA and SMPSA obtain higher values of $k$ as compared to the RMPSA and DMPSA thereby providing the best scheduling performance for qubit-deprived ($\mathcal{C}_{avg}=3$) and qubit-abundant ($\mathcal{C}_{avg}=10.96$) networks. During the demand scheduling procedure, MCSA and SMPSA achieve higher values of $k$ due to considering the qubit contention which occurs during the allocation of paths across unique demands. In MCSA, since the demands are prioritized on the basis of the min-cut, the demand with the least number of paths is accommodated first, thereby leading to a higher $k$. SMPSA however, achieves a higher value of $k$ by allocating paths sequentially, thus avoiding the over-compensation of qubits for the demands.

In Fig. \ref{fig:scheff_avg_phys}, we test the scheduling efficiency (via the value of $k$) on increasing the average physical distance of the links within the network, from $7.44\ km$ to $27.5\ km$. The network environment is setup such that, $\mathcal{C}_{avg} = 9.09$, $|\mathcal{N}| = 100$, and $|\mathcal{D}| = 5$.
Similar to Fig. \ref{fig:scheff_avg_cap}, the MCSA and SMPSA provide the best scheduling performance by achieving a lower rate of decrease in the value of $k$ as compared to the RMPSA and DMPSA.
The decrease in the value of $k$ is due to the decline of $p_s$ on increasing the average physical distance of links, since it results in the decrease of the connectivity of the network corresponding to the decreasing values of $k$ in Fig. \ref{fig:scheff_avg_phys}.
We observe that once again accounting for the qubit-contention across demands, has resulted in a better scheduling efficiency for MCSA and SMPSA.

Fig. \ref{fig:scheff_num_node} illustrates the comparison of the value of $k$ with the increase in the number of nodes ($|\mathcal{N}|$) from $50$ to $250$. We instantiate a quantum network such that $\mathcal{C}_{avg} = 9.09$, the average physical distance of a link within the network is $7.44$, and $|\mathcal{D}| = 5$. As demonstrated in the figure, with the increase in $|\mathcal{N}|$ the MCSA and SMPSA achieve the best performance. The increasing value of $|\mathcal{N}|$ lays emphasis on an efficient scheduling paradigm to obtaining a higher value of $k$. This occurs since a higher number of nodes corresponds to a higher number of paths. However, ensuring the appropriate path allocation of a given demand, and the efficient scheduling of all the demands is essential to achieve the higher $k$ value.

In Fig. \ref{fig:scheff_num_dem}, we examine the variation in the value of $k$ with the increase in the number of demands ($|\mathcal{D}|$) from $5$ to $25$. The quantum network is setup with $\mathcal{C}_{avg} = 9.09$, $\mathcal{N}=100$, and the average physical distance of the links set to $7.44$. As illustrated in the figure, the MCSA and SMPSA experience a lower decrease in the value $k$ in contrast to the RMPSA and DMPSA. With the increase in $|\mathcal{D}|$, the need for qubit-contention free path allocation and efficient scheduling to lower the rate of decrease in the value of $k$ increases.

\begin{figure}[http]
\center
\includegraphics[width=8.4cm]{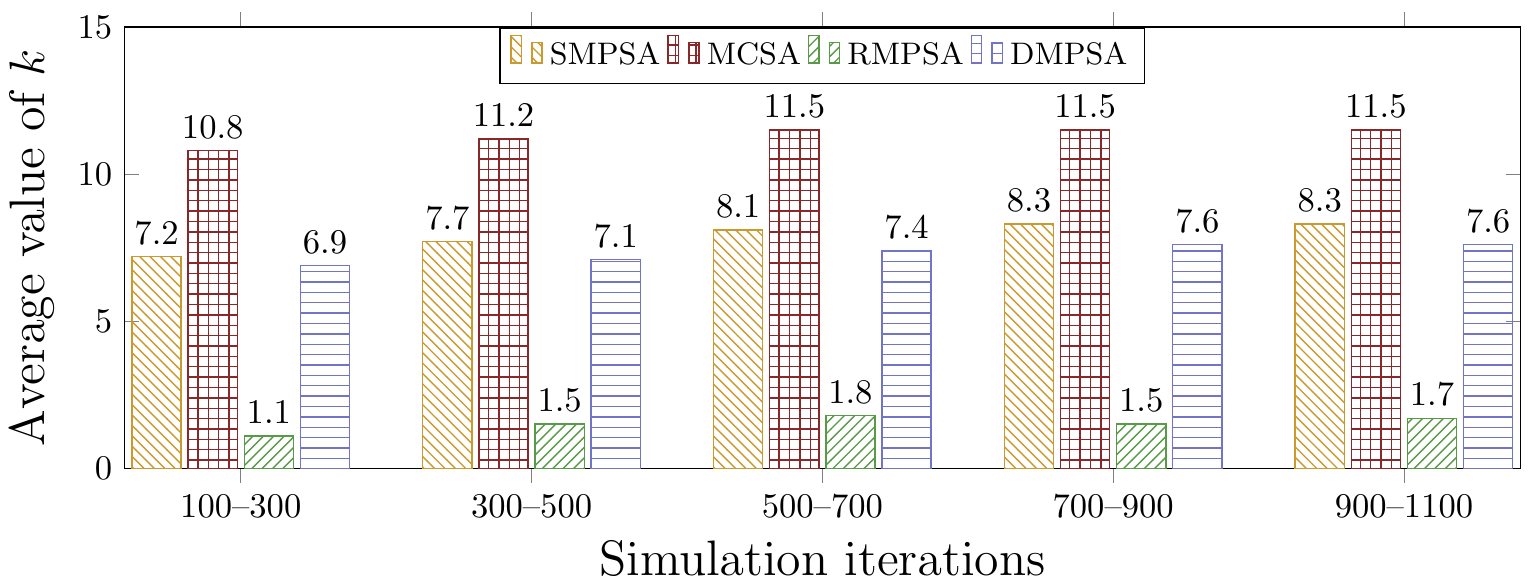}
\caption{The average value of $k$ for the number of simulation iterations varying from $100$ to $1100$.}
\label{fig:ph1-avgkval}
\end{figure}

Clearly from the above descriptions, of the four algorithms being compared in terms of their scheduling efficiency, MCSA and SMPSA achieve the best performance due to taking account of the qubit-contention while scheduling demands. However, Fig. \ref{fig1:ph1-kval} only helps demonstrate the scheduling efficiency of the algorithms relative to the change in various network parameters. To better understand the overall scheduling efficiency of the proposed paradigms, we setup a quantum network with constant parameters such that, the average physical distance of the link is $7.44$, $|\mathcal{N}|=100$, $|\mathcal{D}|$, and $\mathcal{C}_{avg}=9.09$. To deduce the average scheduling performance, we measure the value of $k$ for the algorithms over multiple iterations of distinct randomized topologies. Fig. \ref{fig:ph1-avgkval} illustrates the average scheduling performance (via the average value of $k$) over varying intervals of iterations. As shown in the figure, MCSA obtains a higher average value of $k$ relative to SMPSA, RMPSA, and DMPSA. Moreover, with increasing iterations both the MCSA and SMPSA obtain stable average $k$ values in contrast to RMPSA. Therefore, the MCSA and SMPSA obtain the best scheduling efficiency under changing network conditions, and varying randomized topologies.

\subsection{Path Selection Efficiency}\label{sec:sim-patheff}

\begin{figure}[http]
\center
\subfigure[]{\includegraphics[width=4.2cm]{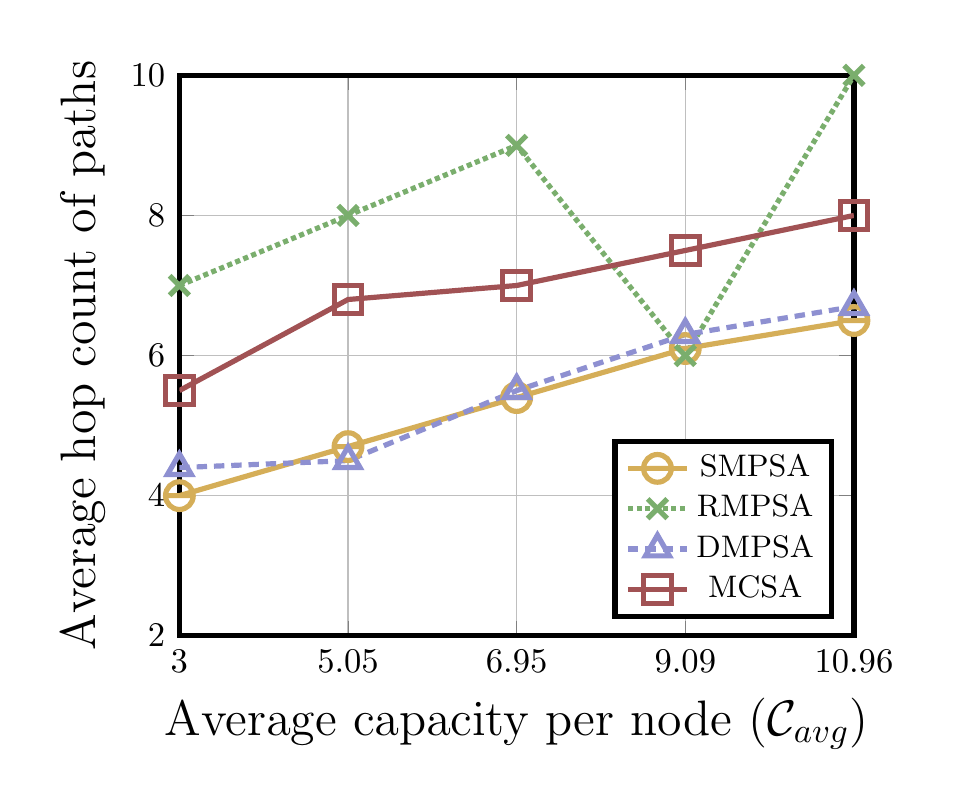} \label{fig:patheff_avg_cap}}
\subfigure[]{\includegraphics[width=4.2cm]{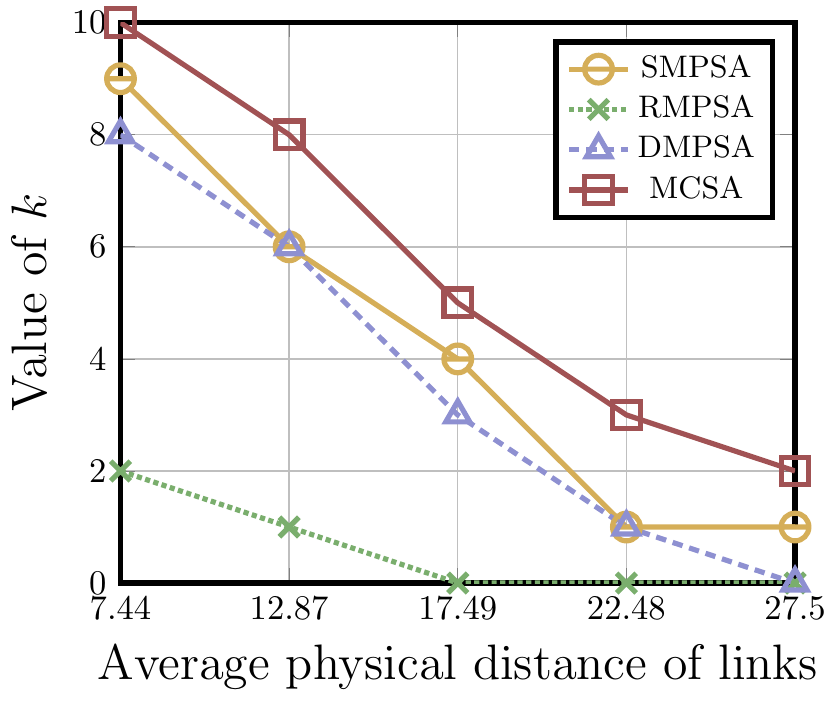}\label{fig:patheff_avg_phys}}
\subfigure[]{\includegraphics[width=4.2cm]{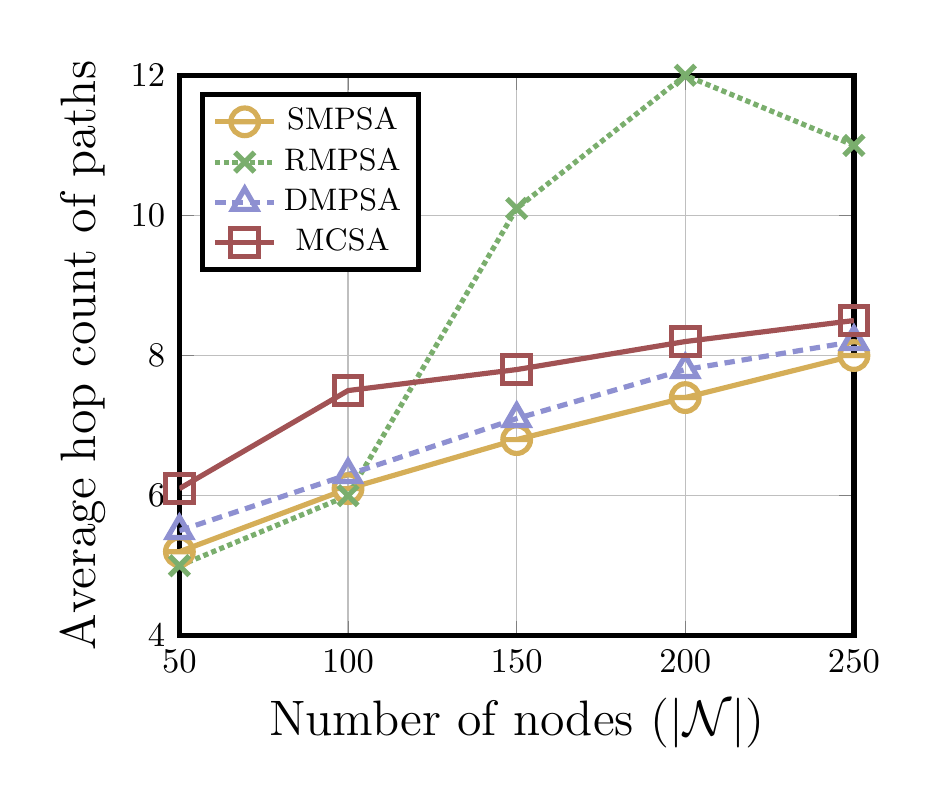} \label{fig:patheff_num_node}}
\subfigure[]{\includegraphics[width=4.2cm]{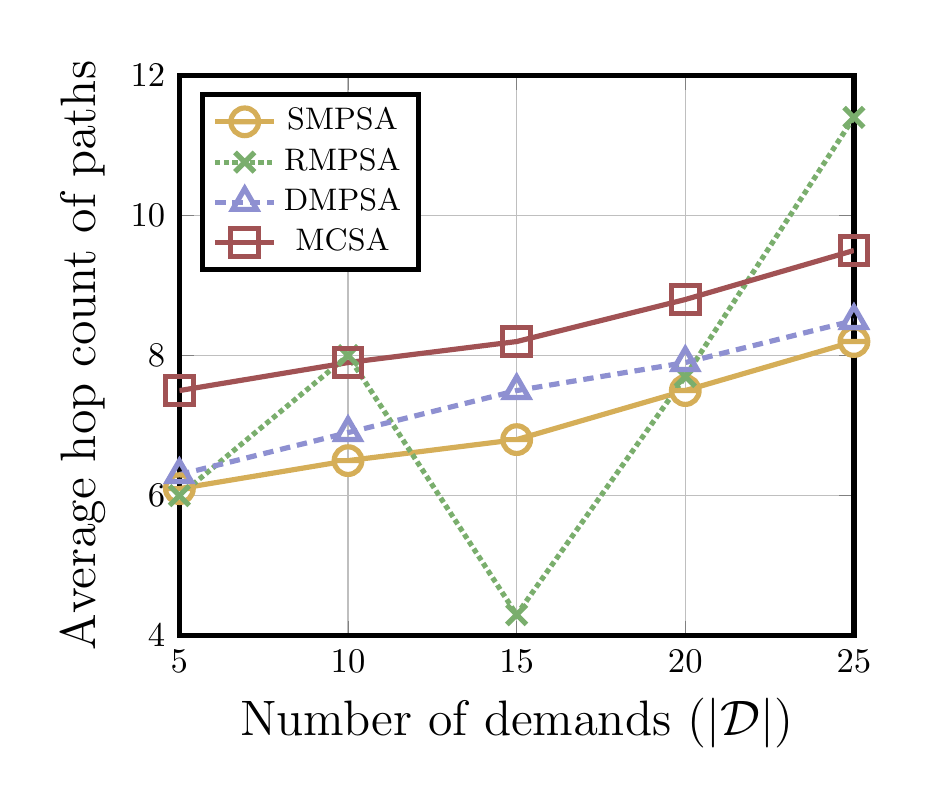}\label{fig:patheff_num_dem}}
\caption{The average hop count of allocated paths when a) the average capacity per node ($\mathcal{C}_{avg}$) ranging from $3$ to $10.96$, b) average physical distance of links ranging from $7.44$ to $27.5$, c) number of nodes ($|\mathcal{N}|$) ranging from $50$ to $250$, and d) number of demands ($|\mathcal{D}|$) ranging from $5$ to $25$.}
\label{fig:ph1-avghopline}
\end{figure}

In this section, we emphasize on establishing the path selection efficiency by measuring the average hop count of the paths allocated through the SMPSA, MCSA, RMPSA, and DMPSA. Fig. \ref{fig:ph1-avghopline} demonstrates the average hop count of the paths allocated by the proposed algorithms with varying network conditions, while Fig. \ref{fig:ph1-avghopbar} examines the average path selection efficiency over randomized topologies.

In Fig. \ref{fig:patheff_avg_cap}, we analyze the average hop count of the allocated paths on increasing the average capacity per node ($\mathcal{C}_{avg}$) from $3$ to $10.96$. In this case, the quantum network is setup identical to that of Fig. \ref{fig:scheff_avg_cap}. As presented in the figure, the MCSA and SMPSA achieve a steady increase in the average hop count of allocated paths as compared to the RMPSA. We observe that the steady increase occurs primarily due to the increasing value of $k$ (as illustrated in Fig. \ref{fig:scheff_avg_cap}) and the total number of paths allocated $\sum\limits_{(src, dst) \in \mathcal{D}}|\mathcal{P}_{(src, dst)}|$. Since the MCSA achieves the highest value of $k$ on increasing $\mathcal{C}_{avg}$, this results in the paths allocated via the MCSA acquiring a higher average hop count in comparison to the SMPSA and DMPSA. Therefore, we remark that the MCSA and SMPSA achieve the best path selection efficiency relative to the scheduling efficiency of the proposed paradigms.

Fig. \ref{fig:patheff_avg_phys} shows the comparison of the average hop count of the paths allocated by the proposed algorithms when the average physical distance of the links is varied from $7.44\ km$ to $27.5\ km$. The quantum network is instantiated similar to that of Fig. \ref{fig:scheff_avg_phys}. Similar to Fig. \ref{fig:patheff_avg_cap}, the MCSA and SMPSA achieve the best path selection efficiency (relative to the scheduling efficiency obtained), through a stable rate of decrease in the average hop count of the allocated paths. The rationale behind the decline of the average hop count of paths allocated is the minimization of the number of entangled qubits (due to the decline in $p_s$), thereby reducing the number of paths allocated by the proposed algorithms. We observe, that the MCSA and SMPSA obtain higher path selection efficiency results due to the prioritization of minimum hop paths, along with the scheduling performance.

In Fig. \ref{fig:patheff_num_node}, we investigate the average hop count of the allocated paths on increasing the number of nodes (|$\mathcal{N}$|) from $50$ to $250$. We initialize the quantum network identical to Fig. \ref{fig:scheff_num_node}. As demonstrated in the figure, the MCSA and SMPSA obtain better path selection efficiency as compared to the RMPSA and DMPSA, relative to the achieved scheduling efficiency. We note that the increase in the average hop count of the paths allocated through the MCSA and SMPSA with the increase of $\mathcal{N}$ occurs due to the allocation of higher number of total paths as compared to RMPSA and DMPSA, as a result of the minimum hop count based path selection.

Fig. \ref{fig:patheff_num_dem} illustrates the comparison of the average hop count of the allocated paths while increasing the number of demands ($|\mathcal{D}|$) from $5$ to $25$. The network parameters are set up similar to Fig. \ref{fig:scheff_num_dem}. Once again, the MCSA and SMPSA achieve the best path selection efficiency relative to scheduling performance of the proposed algorithms. The increase in the average hop count is because of the increase in the paths allocated with the rise of $|\mathcal{D}|$. The MCSA and SMPSA achieve the best path selection efficiency due to being able to accommodate greater number of paths as compared to the RMPSA and DMPSA (covered in detail within \refsec{sim-reseff}), and doing so on the basis of the hop count of paths.

\begin{figure}[t]
\center
\includegraphics[width=8.4cm]{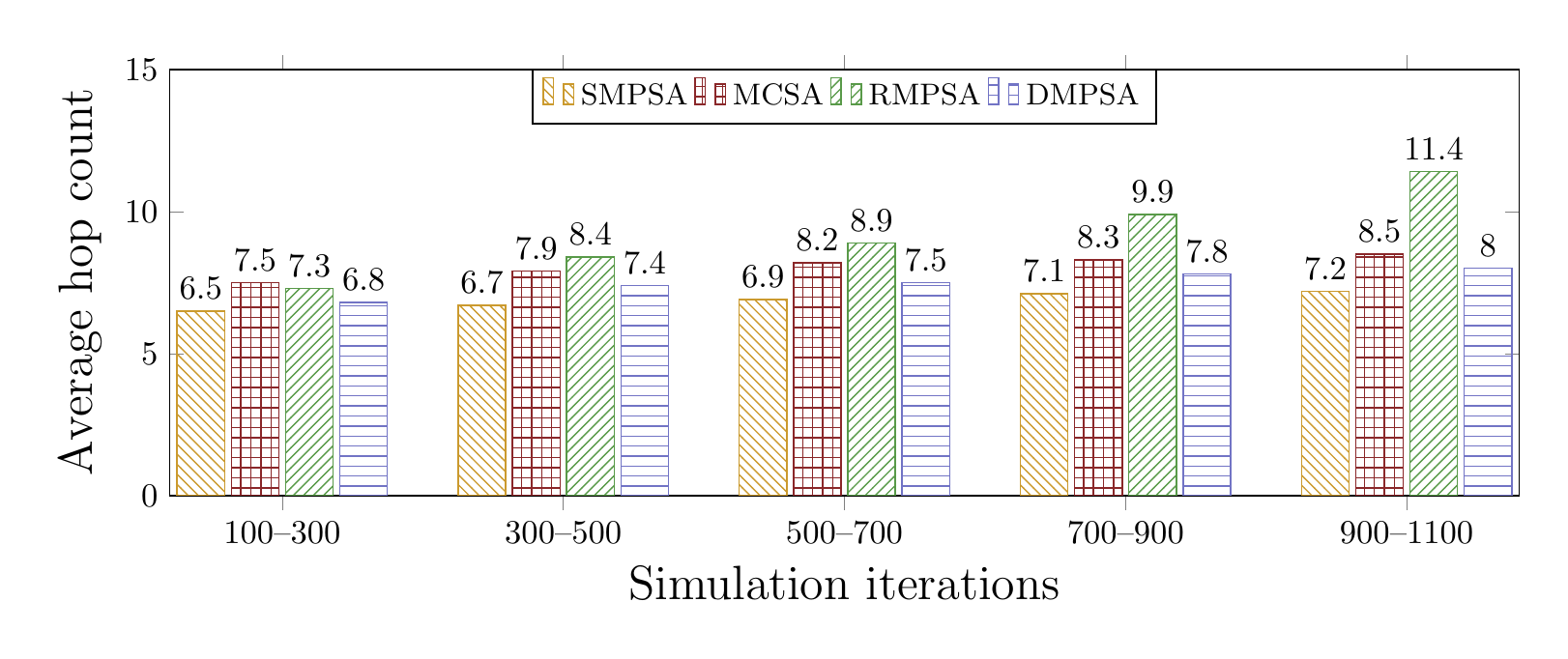}
\vspace{-5pt}
\caption{The average hop count of allocated paths for the number of simulation iterations varying from $100$ to $1100$.}
\vspace{-15pt}
\label{fig:ph1-avghopbar}
\end{figure}

From the analysis of the above figures, we observe that the MCSA and SMPSA attain the best path selection efficiency due to the prioritization of minimum hop paths, and the scheduling efficiency. However, to fully grasp the overall path selection performance of the proposed algorithms, we perform further simulations to evaluate the average hop count of the allocated paths over varying number of simulation iterations. In this case, the quantum network is setup similar to Fig. \ref{fig:ph1-avgkval}. In Fig. \ref{fig:ph1-avghopbar},  we instantiate the quantum network identical to Fig. \ref{fig:ph1-avgkval} and demonstrate the overall path selection performance of the proposed paradigms. As shown in the figure, the SMPSA achieves the lowest average hop count, while the RMPSA obtains the highest, compared to the MCSA and DMPSA. Furthermore, on account of the scheduling efficiency of the proposed algorithms, the MCSA and SMPSA achieve the best overall path selection efficiency taking the total number of paths allocated into consideration. Thus, the MCSA and SMPSA attain the best path selection efficiency under dynamic network environments, and distinct randomized topologies.

\subsection{Resource Utilization Efficiency}\label{sec:sim-reseff}

Finally, we investigate the resource utilization efficiency of the proposed algorithms by measuring the ratio of qubits exhausted in $\mathcal{G}_e$. We note, that the exhausted qubits in $\mathcal{G}_e$ is given by the ratio $(\frac{\mathcal{C}_\mathcal{N} - \mathcal{C}_\mathcal{N}'}{\mathcal{C}_\mathcal{N}})$, such that $\mathcal{C}_\mathcal{N}$ represents the total capacity of the network before path allocation, while $\mathcal{C}_\mathcal{N}'$ reflects the total capacity of the network after path allocation. Fig. \ref{fig1:ph1-remres} examines the qubit depletion ratio under changing network conditions, while Fig. \ref{fig:ph1-avgremres} describes the overall qubit depletion ratio over unique randomized toplogies.

\begin{figure}[http]
\center
\subfigure[]{\includegraphics[width=4.2cm]{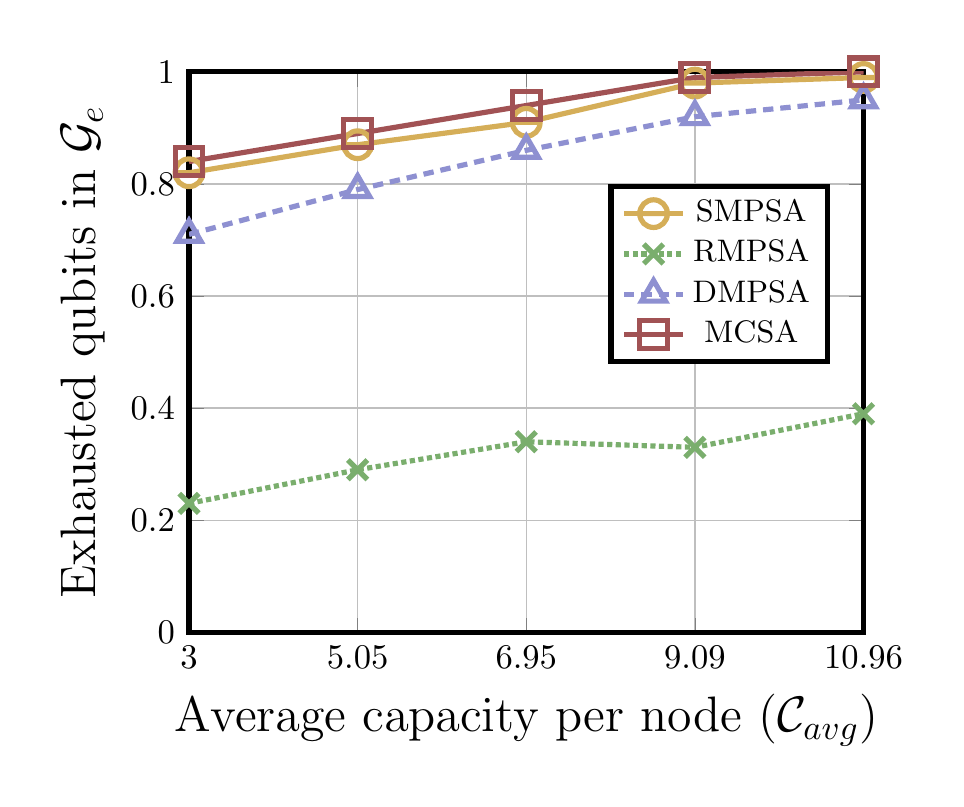} \label{fig:reseff_avg_cap}}
\subfigure[]{\includegraphics[width=4.2cm]{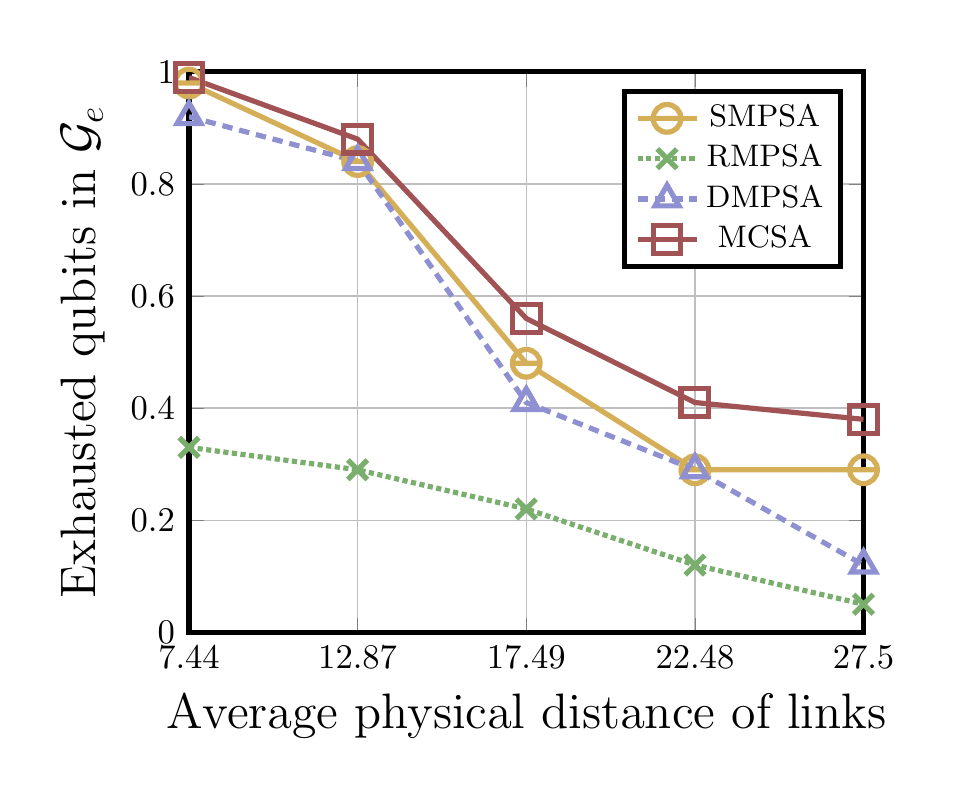}\label{fig:reseff_avg_phys}}
\subfigure[]{\includegraphics[width=4.2cm]{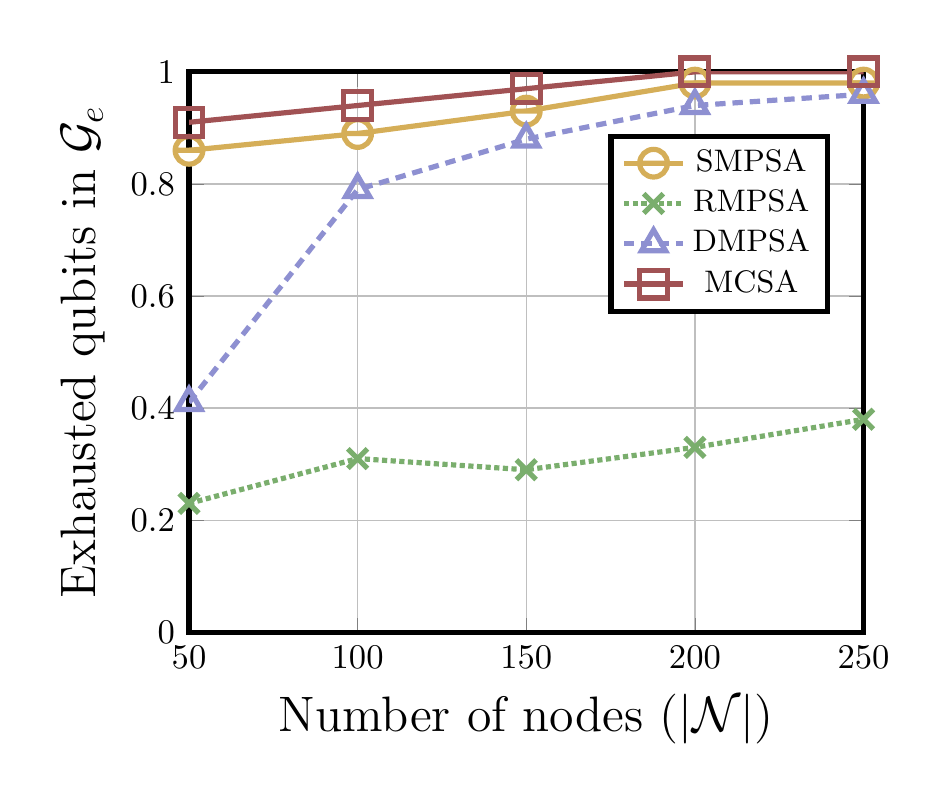} \label{fig:reseff_num_node}}
\subfigure[]{\includegraphics[width=4.2cm]{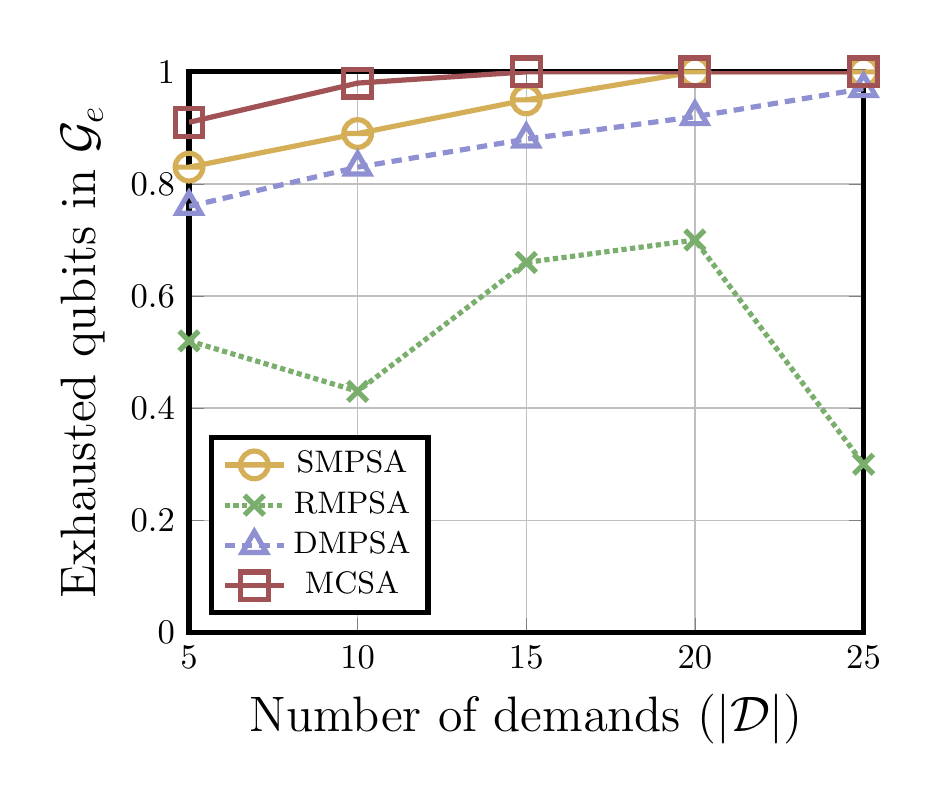}\label{fig:reseff_num_dem}}
\caption{The ratio of exhausted qubits in $\mathcal{G}_e$ when a) the average capacity per node ($\mathcal{C}_{avg}$) ranging from $3$ to $10.96$, b) average physical distance of links ranging from $7.44$ to $27.5$, c) number of nodes ($|\mathcal{N}|$) ranging from $50$ to $250$, and d) number of demands ($|\mathcal{D}|$) ranging from $5$ to $25$.}
\label{fig1:ph1-remres}
\end{figure}

Fig. \ref{fig:reseff_avg_cap} describes the ratio of exhausted qubits in $\mathcal{G}_e$ on varying the average capacity per node ($\mathcal{C}_{avg}$) from $3$ to $10.96$. We note that the quantum network is setup similar to Fig. \ref{fig:scheff_avg_cap}. As shown in the figure, the MCSA and SMPSA achieve the best resource utilization efficiency through obtaining the highest qubit depletion ratio on increasing the average capacity per node. We observe that the increase in the qubit depletion ratio is on account of the availability of a wider selection of paths for all the demands. On the basis thereof, we note that the exceptional performance of the MCSA and SMPSA is primarily due to the allocation of a greater number of paths in comparison to the RMPSA and DMPSA. Since minimum hop paths are preferred by the MCSA and SMPSA, each demand is allocated the minimum number of qubits thereby resulting in a higher number of path allocations, and better resource utilization efficiency.

In Fig. \ref{fig:reseff_avg_phys}, we assess the resource utilization performance of the proposed paradigms, by analyzing the ratio of exhausted qubits for networks with the average physical link distance varying from $7.44\ km$ to $27.5\ km$. In this case, we note that the quantum network is initialized similar to Fig. \ref{fig:scheff_avg_phys}. As demonstrated in the figure, the MCSA and SMPSA maintain the highest qubit depletion ratio with the increase of the average physical distance of links, in contrast to the RMPSA and DMPSA. The decreasing rate of the qubit depletion ratio is principally because of the reduction in the number of entangled qubits (due to the decrease of $p_s$). We observe once more, that the MCSA and SMPSA achieve the best resource utilization efficiency due to the allocation of the minimum number of qubits to a given ordered set of demands.

Fig. \ref{fig:reseff_num_node} demonstrates the resource utilization efficiency of the proposed algorithms by evaluating the qubit depletion ratio while varying the number of nodes (|$\mathcal{N}$|) from $50$ to $250$. The quantum network in this case, is instantiated identical to Fig. \ref{fig:scheff_num_node}. As illustrated within the figure, the MCSA and SMPSA once again achieve the highest qubit depletion ratio with the increase in the number of nodes. The increasing rate of the qubit depletion ratio is primarily due to the increased number of possible paths for all the demands. Moreover, the MCSA and SMPSA obtain the best resource utilization efficiency because of the prioritization of paths consuming minimum qubits. We note that the scheduling efficiency achieved by the MCSA and SMPSA further enhances the resource utilization performance due to the coupling between demand scheduling and qubit allocation for optimal routing.

In Fig. \ref{fig:reseff_num_dem}, we evaluate the resource utilization efficiency of the proposed algorithms, by investigating the qubit depletion ratio for networks with the number of demands ($|\mathcal{D}|$) increasing from $5$ to $25$. The quantum network is setup identical to Fig. \ref{fig:scheff_num_dem}. As shown in the figure, the MCSA and SMPSA maintain the highest qubit depletion ratio on the increase of the number of demands, in comparison to the RMPSA and DMPSA. Once again, we observe that the selection of minimum hop paths yields the highest resource utilization for the MCSA and SMPSA, among the proposed paradigms. Moreover, in the superior scheduling performance of the MCSA helps provide a greater resource utilization efficiency than that of the SMPSA.

\begin{figure}[t]
\center
\includegraphics[width=8.4cm]{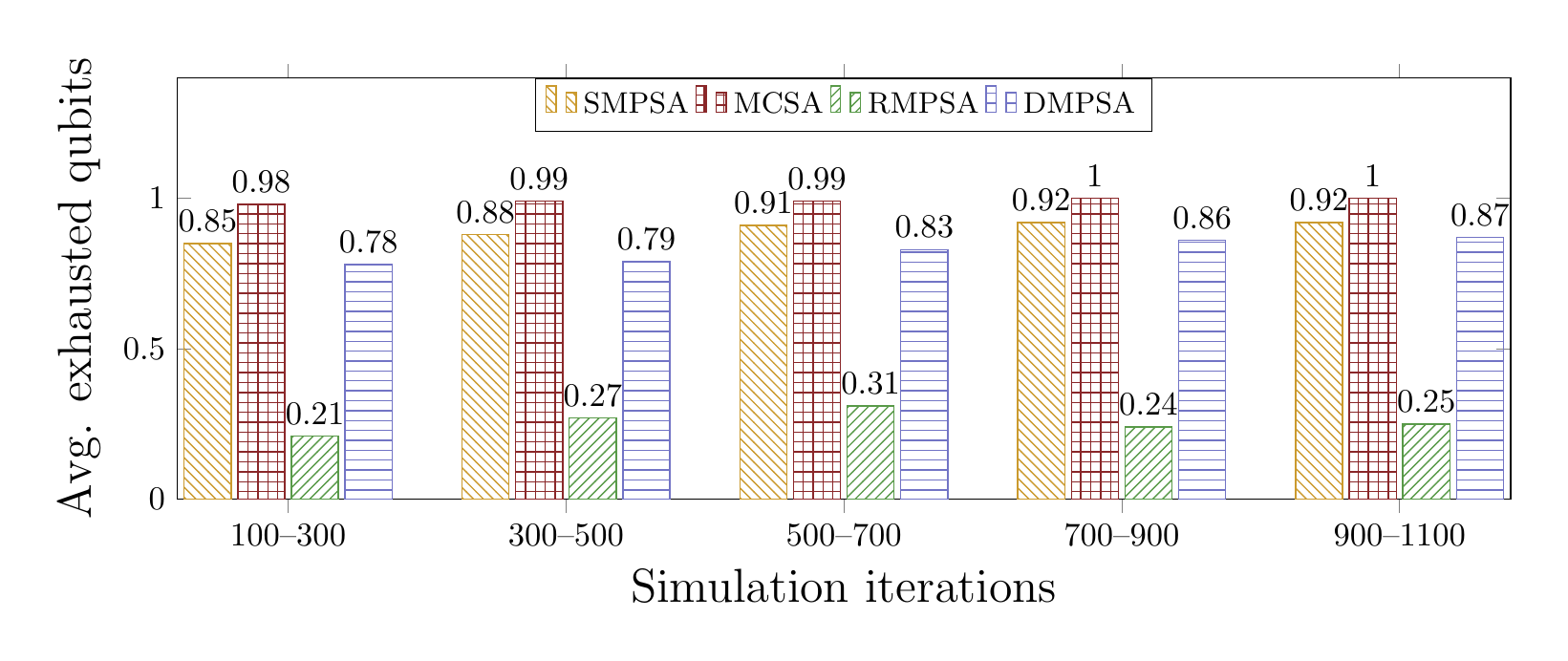}
\caption{The ratio of exhausted qubits in $\mathcal{G}_e$ for the number of simulation iterations varying from $100$ to $1100$.}
\label{fig:ph1-avgremres}
\end{figure}

Evidently, from the investigation of the proposed algorithms the MCSA and SMPSA achieve the best resource utilization efficiency in comparison to the RMPSA and DMPSA, for varying network conditions. We now analyze the overall resource utilization efficiency of the proposed algorithms for varying simulation iterations. We instantiate the quantum network identical to Fig. \ref{fig:ph1-avgkval}. As described in Fig. \ref{fig:ph1-avgremres}, the MCSA obtains the highest overall qubit depletion ratio for each of the iterations, in contrast to the SMPSA, RMPSA, and DMPSA. Furthermore, the MCSA and SMPSA obtain a stable qubit depletion ratio unlike the RMPSA. Thus, from the evaluation of the resource utilization efficiency of the proposed algorithms under dynamic network conditions, and the overall efficiency for a stable network, the MCSA and SMPSA attain the best performance.

\section{Conclusion and Future Works}\label{sec:conclusion}
In this work, we aim to optimize the traffic flexibility in quantum networks by studying the problem of $k$-entangled routing. We have presented a new entanglement routing model that advocates and enables $k$-entangled routing paths for all demands (source-destination pairs) in quantum networks. The proposed algorithms (SMPSA and MCSA) enhance the traffic flexibility for all demands in the network by a big margin compared to other methods.
To demonstrate the feasibility of the proposed model and algorithms, we have conducted not only experiments on the quantum network simulator i.e. NetSquid using real-world topologies and traffic matrices but also simulations to test the network's fidelity.

This work will the lay foundation for further research on ``traffic flexibility" of quantum networks via the $k$-entangled routing problem. We look forward to extending future research topics that hinge on the quantum routing problem with the following directions: 1) maximizing the fidelity of quantum networks, 2) enhancing the quantum networks' resiliency, and 3) designing efficient algorithms to maximize the entangled routing distribution rate for quantum networks.

\normalem

\bibliographystyle{IEEEtran}
\bibliography{bib-tn}
\end{document}